\documentclass[12pt]{article}
\usepackage{amsmath}
\usepackage{amssymb}
\usepackage{amsthm}
\usepackage{fullpage}
\usepackage{mathtools}
\usepackage{color}
\usepackage{tikz}
\usepackage{pgfplots}
  \usepackage{grffile}
\usepackage{combelow}
\usepackage{comment}
\usepackage{amsfonts}
\usepackage{color,ulem}
\usepackage{fullpage}
\usepackage{graphicx}%
\usepackage[active]{srcltx}
\usepackage{mathabx}
\usepackage{verbatim}
\usepackage{cancel}
\usepackage{ulem}
\usepackage{longtable}
\usepackage{caption}
\usepackage{bbm, bm} % for the use of \mathbb{1}
\usepackage[official]{eurosym}
\usepackage[breaklinks,colorlinks=true,linkcolor=blue,citecolor=red,backref=page]{hyperref}
\usepackage{titling}
\usepackage[numbers,sort&compress]{natbib}
\numberwithin{equation}{section}
\usepackage{enumitem}
\setlist[enumerate]{leftmargin=.5in}
\setlist[itemize]{leftmargin=.5in}
\usepackage{times}

\usepackage{yhmath}

  \pgfplotsset{compat=newest}
  \usetikzlibrary{plotmarks}
  \usetikzlibrary{arrows.meta}
  \usepgfplotslibrary{patchplots}

\providecommand{\U}[1]{\protect\rule{.1in}{.1in}}

\newtheorem{theorem}{Theorem}[section]

\newtheorem{corollary}[theorem]{Corollary}

\newtheorem{definition}[theorem]{Definition}

\newtheorem{lemma}[theorem]{Lemma}

\newtheorem{remark}[theorem]{Remark}

\def \R{\mathbb R}
\def \E{\mathbb E}
\def \N{\mathbb N}
\def \P{\mathbb P}

\def \AA {\mathbf A}
\def \BB {\mathbf B}

\newcommand{\D}{\mathcal{D}}

\newcommand{\sign}{\mathrm{sign}}
\def\d{\mathrm{d}}
\newcommand{\indic}{\mathbf{1}}

\newcommand{\cG}{\mathcal{G}}
\newcommand{\ri}{\mathrm{i}}
\newcommand{\rj}{\mathrm{j}}
\newcommand{\rk}{\mathrm{k}}

\makeatletter
\newlength{\temp@wip@width}
\newlength{\temp@wip@height}
\newcommand{\wideinvparen}[1]{%
  \vfuzz=30pt% BAD: to remove overfull vbox warnings...
  \setlength{\temp@wip@width}{\widthof{$#1$}}%
  \setlength{\temp@wip@height}{\heightof{$#1$}}%
  #1\hspace{-\temp@wip@width}%
  \raisebox{\temp@wip@height+1pt}[\heightof{$\wideparen{#1}$}]%
    {\rotatebox[origin=c]{180}{\vbox to 0pt{\hbox{$\wideparen{\hphantom{#1}}$}}}}%
}

\usepackage{accents}
\newlength{\dhatheight}

\usepackage{authblk}
\newcommand{\email}[1]{\thanks{\texttt{#1}}}
\usepackage{abstract}
\begin{document}
\title{A PDE approach for the invariant measure of stochastic oscillators with hysteresis}
\author[1]{Lihong Guo\email{lihguo2223@cityu.edu.cn}}
\author[2]{Harry L. F. Ip\email{harryip2-c@my.cityu.edu.hk}}
\author[2]{Mingyang Wang\email{mywang43-c@my.cityu.edu.hk}}
\affil[1]{\small\it City University of Hong Kong Shenzhen Research Institute, Shenzhen, China}
\affil[2]{\small\it Department of Mathematics, City University of Hong Kong, Hong Kong, China}

\date{ }
\maketitle
\vspace{-5em}
\begin{abstract}
This paper presents a PDE approach as an alternative to Monte Carlo simulations 
for computing the invariant measure of a white-noise-driven bilinear oscillator 
with hysteresis. This model is widely used in engineering to represent highly nonlinear dynamics, 
such as the Bauschinger effect. 
The study extends the stochastic elasto-plastic framework of Bensoussan et al. 
[SIAM J. Numer. Anal. 47 (2009), pp. 3374--3396] 
from the two-dimensional elasto-perfectly-plastic oscillator to the three-dimensional bilinear elasto-plastic oscillator. 
By constructing an appropriate Lyapunov function, 
the existence of an invariant measure is established. 
This extension thus enables the modelling of richer hysteretic behavior 
and broadens the scope of PDE alternatives to Monte Carlo methods. 
Two applications demonstrate the method's efficiency: 
calculating the oscillator's threshold crossing frequency (providing an alternative to Rice's formula) and probability of serviceability. 
\end{abstract}

\section{Introduction}\label{intro}
In this paper, we study the following system of the white-noise-driven bilinear elasto-plastic oscillator (BEPO)  
\begin{equation}\label{Intro_BEPO}
	\ddot{X}(t) + \mathbf{F}(t) 
	= \mathfrak{f}(X(t),\dot{X}(t))+ \sigma \dot{W}(t), 
\end{equation}
with initial condition $(X(0),\dot{X}(0)) = (x,y)\in\R^{2}$. 
Here, $X(t) \in \R$ is the state variable at time $t$, 
dots denote derivatives with respect to $t$, 
$\mathbf{F}(t)$ is the bilinear hysteretic restoring force,  
$\mathfrak{f}: \R^{2} \rightarrow \R$ is a Lipschitz function,  
representing all the non-elastoplastic deterministic forces, 
$\dot{W}$ represents a white noise random forcing where $W$ is a Wiener process, 
and $\sigma$ is the noise intensity.  
Our work is twofold: first, to establish the existence of an invariant measure for system (\ref{Intro_BEPO}) 
under reasonable assumptions on $\mathfrak{f}$, 
and second, to propose a computational method for evaluating its measure, 
based on its backward Kolmogorov equation, 
as an alternative to classical Monte Carlo approaches.

In stochastic structural mechanics,  
the white-noise-driven BEPO is a fundamental model for elastoplastic materials under random excitation 
\cite{C60,NM96,R78,S79,YH87}. 
It serves as an extension of the elasto-perfectly-plastic oscillator (EPPO), the simplest and most widely studied model 
\cite{F08,BT08,BT10,BMPT09,FM12,BJFMM12,BM12,bensoussan2012long,JFMY14,BFMY15,lauriere2015penalization,
	mertz2015degenerate,BMY16,LM16,lauriere2019free, MSW19,lauriere2019penalization,mertz2019numerical,
	mertz2024exponential,ip2025control}. 
The key distinction between these two models lies in the restoring force $\mathbf{F}(t)$. 
For the EPPO, $\mathbf{F}(t)$ depends linearly on elastic deformation, 
and the system state is fully described by its velocity and elastic deformation. 
In contrast, 
BEPO exhibits a bilinear dependence on both elastic and total deformation. 
The latter becomes an additional governing variable in its dynamics, 
extending the EPPO framework to a three-dimensional state space. 
The constitutive models are shown in Fig. \ref{fig:model}.  
\begin{figure}[h]
	\centering
	\captionsetup{width=14cm}
	\captionsetup{font= footnotesize}
	\includegraphics[width=5.5in]{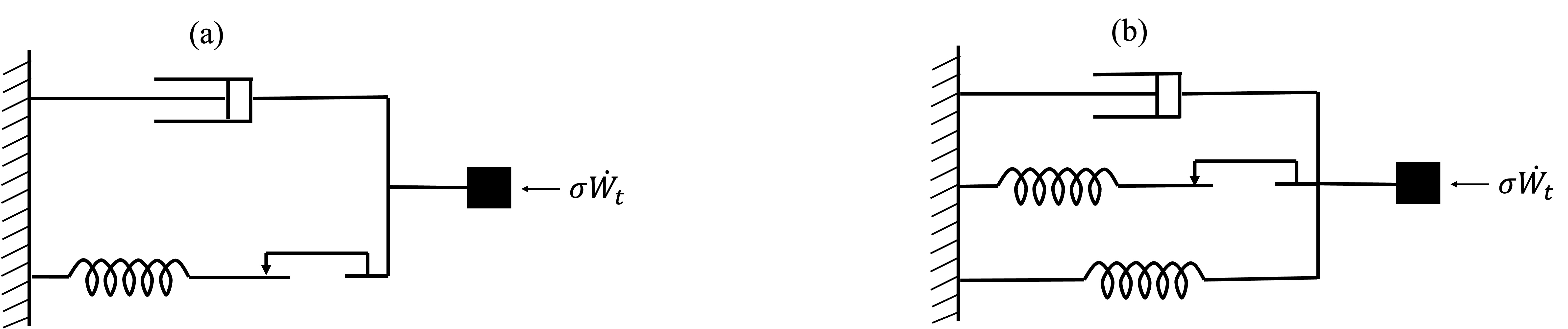}
	\caption{\footnotesize Constitutive models. 
		(a) An EPPO model that contains a linear mass, a dashpot, and a spring connected in series with a Coulomb friction-slip joint.
		(b) A BEPO model, with a spring connected in parallel based on the EPPO model. 
		The two models are excited by a time-dependent random force $\sigma \dot{W}_{t}$, 
		where $\dot{W}_t$ is a white noise, $\sigma$ is the noise intensity.} 
	\label{fig:model}
\end{figure}

The invariant probability measure of stochastic dynamical systems is 
central to assessing structural reliability under random excitations, 
especially in systems that exhibit inelastic behavior.  
A foundational contribution was made in \cite{F08,BT08,BT10} through a stochastic variational inequality (SVI) framework 
for the white-noise-driven EPPO system. 
This formulation obviates the need for phase-wise descriptions of trajectories, 
thereby facilitating the proofs of the existence and uniqueness of the invariant measure.

Building on this foundation, 
subsequent research has extended the theory in several directions. 
New analytical proofs for the existence and uniqueness of the invariant measure were presented in \cite{BM12,mertz2024exponential},  
while an approximate solution to the SVI was developed in \cite{BJFMM12,JFMY14,lauriere2015penalization}. 
Further properties of the stochastic EPPO were derived in \cite{bensoussan2012long,BFMY15}, 
showing that the variance of plastic deformation grows linearly over time, 
and in \cite{BFMY15}, which characterized its probability distribution analytically. 
Applied studies include asymptotic failure risk estimates \cite{LM16} and 
the use of Rice's formula for plastic deformation statistics \cite{FM12}. 
The model has been extended to the systems driven by filtered white noise \cite{mertz2015degenerate,BMY16} and 
colored noise \cite{lauriere2019penalization,ip2025control}, 
with the invariant measure similarly characterized.

However, two significant theoretical challenges remain. 
(1) The existence and uniqueness of a steady-state for the aforementioned 
white-noise-driven BEPO system have not been mathematically established.
Nonetheless, it is empirically supported by stochastic simulations \cite{ip2025control}. 
(2) The curse of dimensionality. 
Numerical computation of invariant measures typically uses either Monte Carlo simulation or PDE methods. 
Common numerical approaches for solving PDEs include, 
but are not limited to, 
the finite difference method \cite{lauriere2019free,MSW19,zhong2024analysis}, 
the finite element method \cite{BMPT09,mertz2019numerical,qiu2020note,gao2021pointwise}, 
the discontinuous Galerkin method \cite{li2022convergent}, 
and the finite volume method \cite{wu2023unconditionally}. 
However, the PDEs studied in this paper contain a small parameter $\lambda$, 
which introduces a boundary layer of width $O(\lambda)$ and consequently requires extremely fine meshes. 
In three dimensions, this leads to prohibitively high computational costs, 
exceeding current practical limits.

In this paper, 
we introduce an SVI modelling system (\ref{Intro_BEPO}) and 
support the existence of a steady state by constructing a Lyapunov function. 
We further develop a numerical PDE method as an alternative to probabilistic simulations. 
This approach is particularly suitable for materials science researchers interested in stationary statistics of elastoplastic models. 
The steady-state analysis offers a profound insight into the system's long-term dynamics, 
enabling the extraction of key quantitative metrics, 
such as the threshold-crossing frequencies and 
the limit-state serviceability probability.

The paper is organized as follows. 
The physical and mathematical formulations of the model are discussed in Section \ref{sec:02}. 
The theoretical analysis of the steady-state is presented in Section \ref{sec:invariantmeasure}. 
Section \ref{sec:num-method} describes the numerical method. 
Two illustrative examples and their numerical results are provided in  Section \ref{sec:num-results}, 
followed by the conclusions in Section \ref{sec:conclusions}.

\section{Formulations of the model}\label{sec:02}
To illustrate the BEPO model formulations from both physical and mathematical perspectives, 
we introduce the following notation. 
Let 
\begin{align*}
	\D = \R\times \R\times (-b,b), \hspace{4em}&\bar{\D} = \R\times \R\times [-b,b],\\
	\D^{-} = \R\times \R\times \{-b\},   \hspace{4em}&\D^{+} = \R\times \R\times \{b\},\quad b>0,\\
	\mathbf{X}(t) = (X(t), Y(t), Z(t)),   \hspace{2em}&\mathbf{x} = (x, y, z). 
\end{align*}
Additionally, let $C_{b}(\bar{\mathcal{D}})$ represent the space of continuous and bounded functions on 
$(\bar{\mathcal{D}}, \mathcal{B})$, 
where $\mathcal{B}$ denotes the Borel $\sigma$-field on $\bar{\mathcal{D}}$. 

\subsection{Physical formulation of the BEPO}\label{phys:BEPO}
As mentioned before, the white-noise-driven BEPO is formally described by 
\begin{align}
	\ddot{X}(t) + \mathbf{F}(t)
	= \mathfrak{f}(X(t),\dot{X}(t))+ \sigma \dot{W}(t), \label{bepo1}
\end{align}
with the initial displacement and velocity 
$X(0) = x$ and $\dot{X}(0) = y$ respectively. 
Here,  
$\dot{W}$ denotes a white noise random forcing in 
the sense that $W$ is a Wiener process, $\sigma>0$ is the noise intensity, and
$\mathfrak{f}(x,y)$ is a Lipschitz function from $\R^{2}$ to $\R$,  representing all non-elastoplastic deterministic forces.
The restoring force $\mathbf{F}(t)$ is a functional that depends on 
the entire trajectory $\{ X(s), 0 \leq s \leq t \}$ up to time $t$, 
its non-linearity comes from the switching of regimes from 
an elastic phase to a plastic one, or vice versa. 
In this model, the irreversible (plastic) deformation $\Delta$ and 
the reversible (elastic) deformation $Z$ at time $t$ satisfy  

\begin{eqnarray*}
	\d Z(t) 
	& = \d X (t), \quad \d \Delta(t) 
	& = 0, \quad \mbox{in elastic phase},\\
	\d \Delta(t) 
	& = \d X (t), \quad \d Z(t) 
	& = 0, \quad \mbox{in plastic phase},
\end{eqnarray*}
while $X(t) = Z(t) + \Delta(t)$. 
Typically, plastic phase occurs when $ \vert Z(t) \vert = b$ and 
elastic phase occurs when $\vert Z(t) \vert < b$, 
here $b>0$ is an elasto-plastic bound. 
We will consider the case in which the force $\mathbf{F}(t)$ is 
a linear combination of the elastic and plastic deformations. 
This force is a linear function of $Z$ 
(while $\Delta = X -Z$ remains constant) in the elastic phase and 
a linear function of $X$ (while $Z$ remains constant at $\pm b$) in the plastic phase as follows:
\begin{equation}\label{bepo2}
	\mathbf{F}(t) 
	= k (1-\alpha) Z(t) + k \alpha X(t), \quad 0 \leq \alpha \leq 1, \quad k > 0.
\end{equation} 
Note that when $\alpha = 0$, the system \eqref{bepo1}--\eqref{bepo2} reduces to the classical EPPO model.

\subsection{Mathematical formulation of the BEPO}
Consider the triple 
$\{ (X(t),Y(t),Z(t)) , \: t \geq 0 \}$ 
solution of the following SVI
\vspace{-0.5em}
\begin{equation}\label{SVIBEPO}
	\begin{cases}
		\d X(t) =  Y(t)\d t,\\
		\d Y(t) = [\mathfrak{f}(X(t),Y(t)) - k (1-\alpha) Z(t) - k \alpha X(t)]\d t + \sigma  \d W(t),\\
		(\d Z(t) - Y(t)\d t)(\xi - Z(t)) \geq 0, ~\forall |\xi| \leq b, ~|Z(t)| \leq b,  
	\end{cases}
\end{equation}
with initial condition 
$(X(0),Y(0),Z(0)) = (x,y,z) \in \bar{\D}$.
The SVI \eqref{SVIBEPO} is a concise and rigorous way to express the elasto-plastic dynamics, 
shown in Section \ref{phys:BEPO}, as it governs the alternance between elastic and plastic regimes.
It is a well-posed problem, see \cite{PR14} Chapter 4.

\subsection{Reformulation with stopping times}
Without loss of generality, 
we take $\mathbf{x} \in \D$ as an example to illustrate the process generated by SVI \eqref{SVIBEPO}. 
Recall from \cite{BT08} and define two sets of stopping times as follows: 
\begin{align}
	\begin{cases}
		\theta_{n+1} = \inf\{t>\tau_{n}~ | ~ |Z(t)| = b \},\\
		\tau_{n+1} = \inf\{t > \theta_{n+1}~|~ \sign (Y(t))  = - \sign(Z(t))\},
	\end{cases}\label{stopping_time}
\end{align}
where $\tau_{0} = \theta_{0} = 0$ and $n\in\N_{0}$.  
To show the behavior of the restoring force $\mathbf{F}(t)$ with respect to $X(t)$, 
we formally draw the evolution curve according to the stopping time in  Fig. \ref{fig:phases}.

\begin{remark}
	$(\theta_{n})_{n\geq 1}$ corresponds to each entry in the plastic regime, 
	we use the indicator $|Z(t)| = b$ to describe this behavior.  
	$(\tau_{n})_{n\geq 1}$ corresponds to each exit of the plastic regime. 
	There are at least two indicators that can describe this behavior: 
	$|Z(t)| < b$ and the variation of $\sign(Y(t))$.  We choose to use the latter one in sets \eqref{stopping_time}. 
\end{remark}

\begin{figure}[htpb]
	\centering
	\captionsetup{width=14cm}
	\captionsetup{font= footnotesize}
	\includegraphics[width=6.5in]{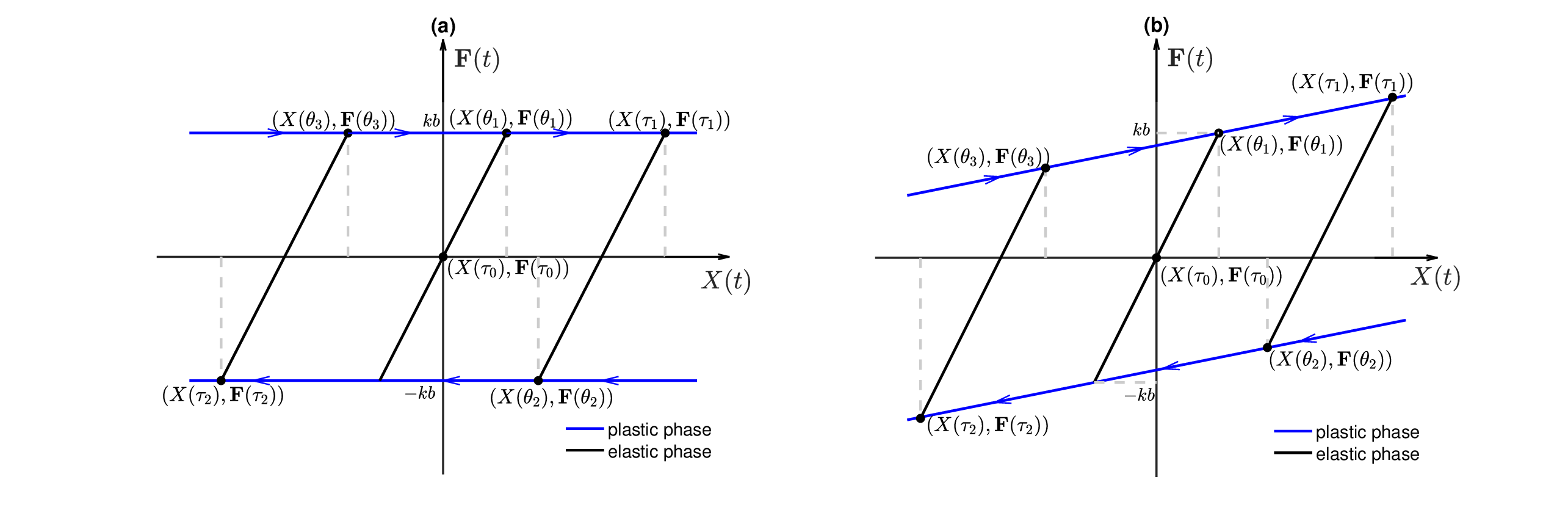}
	\caption{\footnotesize Archetypal evolution of the restoring force $\mathbf{F}(t)$  versus $X(t)$. 
		(a) The EPPO model. In the plastic phase, $\mathbf{F}(t)$ remains a constant.   
		(b) The BEPO model. In the plastic phase,  $\mathbf{F}(t)$ is a linear function. } 
	\label{fig:phases}
\end{figure}

When $t\in(\tau_{n}, \theta_{n+1})$, we have $|Z(t)| <b$ and 
\begin{align*}
	\begin{cases}
		\d X(t) = Y(t)\d t,\\
		\d Y(t) = [\mathfrak{f}(X(t), Y(t)) -k(1-\alpha)Z(t) - k\alpha X(t)]\d t + \sigma \d W(t),\\
		\d Z(t) = Y(t)\d t. 
	\end{cases}
\end{align*}
When $t\in (\theta_{n+1}, \tau_{n+1})$, we have $|Z(t)| = b$, $\sign(Z(t))Y(t) \geq 0$, 
and
\begin{align*}
	\begin{cases}
		\d X(t) = Y(t)\d t,\\
		\d Y(t) = [\mathfrak{f}(X(t), Y(t)) -k(1-\alpha)\sign(Z(t))b - k\alpha X(t)]\d t + \sigma \d W(t),\\
		\d Z(t) = 0. 
	\end{cases}
\end{align*}

\section{Invariant probability measure\label{sec:invariantmeasure}}
In this section, 
we first derive the backward Kolmogorov equations, 
then we reveal the existence of the invariant probability measure. 
Moreover, we provide an approximation to this measure. 
As the proofs are classic, we place them in Appendix \ref{app-A01} for the reader's convenience.

\subsection{Backward Kolmogorov equations}
Let $\beta(x,y,z) = \mathfrak{f} (x, y) - k(1-\alpha) z -k\alpha x$, 
define the operators $\AA$, $\BB_{+}$, and $\BB_{-}$ by
\begin{align*}
	&\AA  \triangleq \dfrac{\sigma^{2}}{2}\partial^{2}_{y}
	+ y \partial_{x}
	+ \beta(x,y,z)\partial_{y}
	+ y\partial_{z},\\
	&\BB_{+} \triangleq \dfrac{\sigma^{2}}{2}\partial^{2}_{y}
	+ y \partial_{x}
	+ \beta(x,y,b)\partial_{y} + \min(0,y)\partial_{z},\\
	&\BB_{-}  \triangleq \dfrac{\sigma^{2}}{2}\partial^{2}_{y}
	+ y \partial_{x}
	+ \beta(x,y,-b)\partial_{y}+ \max(0,y)\partial_{z}. 
\end{align*}

\begin{theorem}\label{Prop01_BEPO}
	Let $\mathbf{X}$ solve SVI \eqref{SVIBEPO}, 
	let $g\in L^{2}(\bar{\D})$, and consider
	\begin{align}\label{SVI-sol01}
		\psi(t, \mathbf{x}) = \E[g(\mathbf{X}(t)) | \mathbf{X}(0) = \mathbf{x}], 
	\end{align}
	then
	\begin{align}
		\begin{cases}\label{SVI-eq01}
			\partial_{t}\psi - \AA \psi = 0,~~\text{in~}(0,T]\times \D,\\
			\partial_{t}\psi - \BB_{\pm} \psi = 0,~~\text{in~}(0,T]\times \D^{\pm},\\
			\psi(0,\mathbf{x}) = g(\mathbf{x}). 
		\end{cases}
	\end{align}
	
\end{theorem}

\subsection{Existence of an invariant probability measure}
\noindent 
By selecting an appropriate Lyapunov function, 
an energy estimate for the SVI \eqref{SVIBEPO} is established in Theorem \ref{prop-03}. 
Subsequently, the existence of an invariant probability measure for   $\mathbf{X}(t)$ is deduced from Theorem \ref{Thm-01}. 
To compute this invariant measure, 
an approximation scheme based on a PDE approach is presented in Theorem \ref{Thm-invariant-measure}.
To show these, the following assumptions are needed. 
For any $x,y\in\R$, 
\begin{align}
	y\mathfrak{f}(x,y) &\leq - c_{0}y^{2} + c_{1}xy + c_{2}, ~\text{with}~c_{0} >0, ~c_{1}\leq k\alpha, ~c_{2}\geq 0,\label{assum-01}\\
	x\mathfrak{f}(x,y) &\leq -d_{0}xy + d_{1} x^{2} + d_{2}, ~\text{with}~d_{0} \geq c_0, ~d_{1}<k\alpha, ~d_{2}\geq 0. \label{assum-02}
\end{align}

\begin{remark}
	Assumptions \eqref{assum-01} and \eqref{assum-02} contain the case $\mathfrak{f}(x,y) = -c_{0}y$ from EPPO 
	\cite{BMPT09},	where $c_{0}>0$. 
\end{remark}

\begin{theorem}\label{prop-03} 
	Let
	$V(x,y) \triangleq  ( k \alpha + \frac{c_{0}d_{0}}{2} - c_{1} )x^2 + y^2 +  c_0 xy$. 
	For any $t>0$, 
	\begin{align*}
		\mathbb{E} [V(X(t),Y(t))] \leq V(X(0),Y(0)) +C/C_{1}, 
	\end{align*}
	where $C = \sigma^{2} + 2c_{2} + c_{0}d_{2} + (kb(1-\alpha))^{2}\left(\frac{2}{c_{0}} + \frac{c_{0}}{2(k\alpha - d_{1})}\right)$ 
	and 
	$C_{1} = \min\left(\frac{c_{0}}{3}, \frac{c_{0}(k\alpha - d_{1})}{2k\alpha +c_{0}d_{0} + c^{2}_{0} - 2c_{1}}\right)$.
\end{theorem}

Next, we define the probability law $\mu(t)$ on $(\bar{\mathcal{D}}, \mathcal{B})$ associated with $\mathbf{X}(t)$ by
\begin{align*}
	\mu(t)(\phi) \triangleq \E[\phi (\mathbf{X}(t))] = \mu(0) (P_{t}\phi), ~~\text{for}~\phi\in	C_{b}(\bar{\mathcal{D}}). 
\end{align*}
The operator $P_{t}$ is defined by
$P_{t}\phi(\mathbf{x}) \triangleq \E[\phi(\mathbf{X}(t)) | \mathbf{X}(0) = \mathbf{x}]$, 
for $\phi\in	C_{b}(\bar{\mathcal{D}})$.

\begin{theorem}\label{Thm-01}
	The process $\mathbf{X}(t)$ generated by SVI \eqref{SVIBEPO} 
	with initial condition $\mathbf{X}(0) = \mathbf{x}\in \bar{\mathcal{D}}$ 
	admits an invariant measure $\mu$, i.e., 
	$\mu(0) = \mu$ implies $\mu(t) = \mu$, $\forall t \geq 0.$  
\end{theorem}

Such an invariant measure $\mu$ for $\mathbf{X}$ 
must annihilate the infinitesimal generator of $P_{t}$ on any test function, 
as shown in the following corollary. 
\begin{corollary}\label{coro02}
	The measure $\mu$ satisfies the ultra-weak formulation
	\begin{align*}
		\int_{\mathcal{D}} \AA\varphi(x,y,z) \d \mu 
		+ \int_{\mathcal{D}^+} \BB_{+}\varphi(x,y,b) \d \mu
		+ \int_{\mathcal{D}^-} \BB_{-}\varphi(x,y,- b) \d \mu
		= 0,~~\forall \varphi\in C_{b}(\bar{\mathcal{D}}). 
	\end{align*}
\end{corollary}

\begin{theorem}\label{Thm-invariant-measure}
	Let $\lambda>0$ and $g\in L^{2}(\bar{\D})$. 
	The unique solution $u_{\lambda}$ bounded by 
	$\frac{||g||_\infty}{\lambda}$ solves the following PDEs
	\begin{align}\label{PDE-ulambda}
		\begin{cases}
			\lambda u - \AA u  = g,~~\text{in~}\mathcal{D}, \\
			\lambda u - \BB_\pm u = g,~~\text{in~}\mathcal{D}^\pm,
		\end{cases}
	\end{align}
	satisfying 
	\begin{align}
		\lim_{\lambda \to 0} \lambda u_\lambda(\mathbf{x})
		= \lim_{t\rightarrow \infty} \E[g(\mathbf{X}(t)) | \mathbf{X}(0) = \mathbf{x}]
		= \int_{\bar{\mathcal{D}}} g \d \mu, ~\forall \mathbf{x}\in\bar{\mathcal{D}}. \label{thm3-01}
	\end{align}
\end{theorem}
\begin{remark}
	As noted in \cite{BMPT09}, 
	\eqref{thm3-01} is numerically significant in that it allows for the calculation 
	of $\lim\limits_{t\rightarrow \infty} \E[g(\mathbf{X}(t)) | \mathbf{X}(0) = \mathbf{x}]$ without solving a time-dependent problem. 
	Furthermore, 
	the resulting limit is independent of the initial position $\mathbf{x}$.
\end{remark}

\begin{remark}
	The uniqueness of the invariant measure remains a challenging issue, 
	as the system exhibits degenerate noise and non-smooth drift due to hysteresis, 
	which precludes the application of standard ergodic theory for diffusions. 
	Nevertheless, in numerical experiments, 
	we compare the ``spatial average" derived from the invariant measure (via PDE methods) with 
	the ``temporal average" obtained from long-term simulations (via Monte Carlo). 
	These comparisons strengthen our confidence in the ergodicity assumption. 
\end{remark}

\section{Numerical method for the invariant measure}\label{sec:num-method} 
For numerical computation, 
the unbounded domain $\bar{\mathcal{D}}$ is  truncated by
$\mathcal{D}_{XY} \triangleq [-\bar{x}, \bar{x}] \times [-\bar{y}, \bar{y}] \times [-b, b]$, 
where $\bar{x}$ 
and $\bar{y}$ are chosen sufficiently large that the probability of finding the 
underlying process outside $\mathcal{D}_{XY}$ is negligible. 
Note that the small parameter $\lambda$ in PDEs \eqref{PDE-ulambda} poses a challenge 
for the finite difference method: 
to preserve numerical accuracy, 
an extremely fine mesh must be employed, 
which incurs high computational costs. 
To address this efficiency bottleneck, we introduce the scaling transformation:
$\tilde{x} = \lambda x$, $\tilde{y} = \lambda y$, $\tilde{z} = \lambda z$, 
and let $v(\tilde{x},\tilde{y},\tilde{z}) \triangleq 
\lambda u(\frac{\tilde{x}}{\lambda},\frac{\tilde{y}}{\lambda},\frac{\tilde{z}}{\lambda})$. 
As a consequence, the truncated domain changed to 
$[-\lambda \bar{x}, \lambda \bar{x}] \times [-\lambda \bar{y}, \lambda \bar{y}] 
\times [-\lambda b, \lambda b]$. 
The numerical approximation is built on the following three-dimensional finite difference grid
\begin{align*}
	\cG \triangleq \left\{(\tilde{x}_{\ri},\tilde{y}_{\rj},\tilde{z}_{\rk}) 
	= (-\lambda\bar{x} + (\ri-1) \delta_{\tilde{x}}, -\lambda\bar{y} + (\rj-1) \delta_{\tilde{y}}, 
	-\lambda b + (\rk-1) \delta_{\tilde{z}})\right\}_{1 \le \ri \le I, 1 \le \rj \le J, 1 \le \rk \le K},
\end{align*}
where $\delta_{\tilde{x}} \triangleq \frac{2\lambda\bar{x}}{I-1}$, 
$\delta_{\tilde{y}} \triangleq \frac{2\lambda\bar{y}}{J-1}$, 
$\delta_{\tilde{z}} \triangleq \frac{2\lambda b}{K-1}$, 
$I,J$, and $K$ are odd integers greater than 1. 
Let $v_{\ri, \rj, \rk}$ denote the numerical approximation of $v(\tilde{x}_{\ri}, \tilde{y}_{\rj}, \tilde{z}_{\rk})$, 
let $\bm{v}$ be a vector whose elements are $v_{\ri, \rj, \rk}$.  
For the sake of completeness, 
we provide the following notations of the difference symbols with respect to $\tilde{x}$, 
analogous definitions apply to $\tilde{y}$ and $\tilde{z}$. 
\begin{itemize}
	\item The first-order forward and backward finite differences:
	
	$(\partial^{\mathrm{f}}_{\tilde{x}}v)_{\ri,\rj,\rk} \triangleq (v_{\ri+1, \rj, \rk} - v_{\ri, \rj, \rk})/\delta_{\tilde{x}}$, 
	$(\partial^{\mathrm{b}}_{\tilde{x}}v)_{\ri,\rj,\rk} \triangleq (v_{\ri, \rj, \rk} - v_{\ri-1, \rj, \rk})/\delta_{\tilde{x}}$. 
	\item  The second-order centered finite difference: 
	
	$(\partial^{2}_{\tilde{x}}v)_{\ri,\rj,\rk} \triangleq (v_{\ri+1, \rj, \rk} - 2v_{\ri, \rj, \rk} +  v_{\ri-1, \rj, \rk})/\delta^{2}_{\tilde{x}}$.  
	\item The second-order forward and backward finite differences: 
	
	$(\partial^{\mathrm{ff}}_{\tilde{x}}v)_{\ri,\rj,\rk} \triangleq (-3v_{\ri, \rj, \rk} + 4v_{\ri+1, \rj, \rk} - v_{\ri+2, \rj, \rk})/(2\delta_{\tilde{x}})$, 
	$(\partial^{\mathrm{bb}}_{\tilde{x}}v)_{\ri,\rj,\rk} \triangleq (3v_{\ri, \rj, \rk} - 4v_{\ri-1, \rj, \rk} + v_{\ri-2, \rj, \rk})/(2\delta_{\tilde{x}})$. 
	\item The second-order upwind difference: 
	
	$(c(\tilde{x},\tilde{y},\tilde{z})\tilde{\partial}_{\tilde{x}}v)_{\ri,\rj,\rk} 
	\triangleq \max(0,c_{\ri,\rj,\rk})(\partial^{\mathrm{ff}}_{\tilde{x}}v)_{\ri,\rj,\rk}
	+ \min(0,c_{\ri,\rj,\rk})(\partial^{\mathrm{bb}}_{\tilde{x}}v)_{\ri,\rj,\rk}$. 
\end{itemize} 
The discretization of the scaled version of PDEs \eqref{PDE-ulambda} follows the form: 
\begin{align}
	v_{\ri,\rj,\rk}  - (L\bm{v})_{\ri,\rj,\rk} = \tilde{g}_{\ri,\rj,\rk}, \label{Disc-01} 
\end{align}
where 
$(L\bm{v})_{\ri,\rj,\rk} \triangleq (L_{\tilde{y}}\bm{v})_{\ri,\rj,\rk} + S_{\tilde{x}}
+S_{\tilde{z}}$,  
with $(L_{\tilde{y}}\bm{v})_{\ri,\rj,\rk} = \frac{\lambda\sigma^{2}}{2}(\partial^{2}_{\tilde{y}}v)_{\ri,\rj,\rk}
+ (\tilde{\beta}(\tilde{x},\tilde{y},\tilde{z})\tilde{\partial}_{\tilde{y}}v)_{\ri,\rj,\rk}$, 
$\tilde{\beta}$ and $\tilde{g}$ are the scaled version of the functions $\beta$ and $g$. 
The terms $S_{\tilde{x}}$ and $S_{\tilde{z}}$ 
are determined by the region types and boundary indices, as specified below:
\begin{itemize}
	\item  Interior (black triangles, shown in Fig. \ref{fig:discretization}):
	$S_{\tilde{x}} = (\frac{\tilde{y}}{\lambda}\tilde{\partial}_{\tilde{x}}v)_{\ri,\rj,\rk}, \quad 
	S_{\tilde{z}} = (\frac{\tilde{y}}{\lambda}\tilde{\partial}_{\tilde{z}}v)_{\ri,\rj,\rk}$. 
	\item Boundaries (blue/red squares and diamonds, shown in Fig. \ref{fig:discretization}):
	\begin{align*}
		S_{\tilde{x}} =  
		\begin{cases}
			\max(0,\frac{\tilde{y}_{\rj}}{\lambda}) (\partial^{\mathrm{ff}}_{\tilde{x}}v)_{\ri,\rj,\rk}, & \text{if}~\ri = 1,\\
			\min(0,\frac{\tilde{y}_{\rj}}{\lambda}) (\partial^{\mathrm{bb}}_{\tilde{x}}v)_{\ri,\rj,\rk}, & \text{if}~ \ri = I, 
		\end{cases}
		\quad\quad 
		S_{\tilde{z}} =  
		\begin{cases}
			\max(0,\frac{\tilde{y}_{\rj}}{\lambda}) (\partial^{\mathrm{ff}}_{\tilde{z}}v)_{\ri,\rj,\rk}, & \text{if}~\rk = 1,\\
			\min(0,\frac{\tilde{y}_{\rj}}{\lambda}) (\partial^{\mathrm{bb}}_{\tilde{z}}v)_{\ri,\rj,\rk}, & \text{if}~ \rk = K.
		\end{cases}
	\end{align*}
\end{itemize}
Additionally, for the gray round points ($y$-boundary), 
homogeneous Neumann conditions are imposed to enforce zero-flux: 
$(\partial^{\mathrm{f}}_{\tilde{y}}v)_{\ri,1,\rk} = 0$ and 
$(\partial^{\mathrm{b}}_{\tilde{y}}v)_{\ri,J,\rk} = 0$.  
The numerical algorithm can be found in Appendix \ref{app-A02}.  

\begin{figure}[h!]
	\centering
	\includegraphics[width=4.5in]{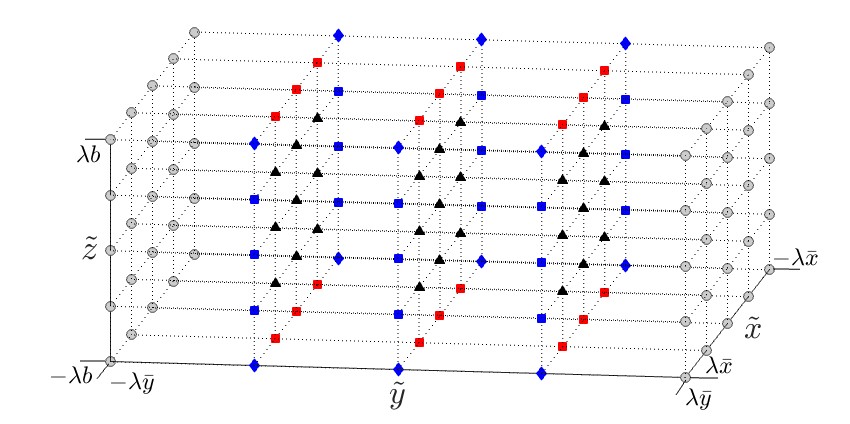}
	\caption{\footnotesize Discretization of $\mathcal{D}_{XY}$. 
		At black triangles, the discretized equation is satisfied. 
		At gray round points, homogeneous Neumann boundary conditions are used. 
		The non-standard boundary conditions are employed at blue squares/diamonds and red squares.}
	\label{fig:discretization}
\end{figure}

\begin{remark}
	Due to the absence of diffusion in $\tilde{x}$ and $\tilde{z}$, 
	outward flux terms in $S_{\tilde{x}}$ and $S_{\tilde{z}}$ are omitted at the boundaries 
	$\tilde{x} = \pm \lambda\bar{x}$ and $\tilde{z} = \pm \lambda b$, 
	ensuring the solution be determined by inward propagation. 
\end{remark}
\begin{remark}
	Note that the second-order upwind scheme involves two adjacent upstream (or downstream) nodes, 
	some special treatment are required for finite difference nodes near the boundary to 
	avoid referencing values outside the computational domain. 
	Drawing on the analysis in \cite{vanka1987second,qiu2023global}, 
	this study reverts to the first-order upwind scheme for such boundary-adjacent points.  
\end{remark}

\section{Numerical results} \label{sec:num-results}
In seismology, statistical physics, and engineering, 
it is necessary to investigate the frequency of threshold crossing and the probability of the risk of failure, 
especially for risk assessment and schematic design.  
For instance,  in building and bridge structural design, 
quantifying the frequency and amplitude of ground motion 
at specific sites is essential to ensure the seismic resilience of constructions. 
﻿
\subsection{Frequency of crossing threshold}
In this subsection, 
we focus on empirical approaches to compute the crossing frequency of SVI \eqref{SVIBEPO} on a unit time interval. 
One of the powerful tools to compute this quantity is the Rice's formula:
\begin{align*}
	\nu(a_{1}) = \int_{-\infty}^{\infty} |y| p(a_{1},y)\d y, 
\end{align*}
where $p(\cdot,\cdot)$ is the joint density of $X(0)$ and $\dot{X}(0)$, 
$a_{1}\in\R$ is a given threshold. 
For details, refer to Appendix \ref{app-A03} and \cite{LLR83}.  
Generally, the exact expression of the invariant probability density function of  \eqref{SVIBEPO} is not easy to obtain.    
Therefore, we cannot use Rice's formula directly.   
Instead, we provide an approximation to this quantity. 
Specifically,  
Theorem \ref{Thm-invariant-measure} implies that 
\begin{align*}
	\nu(a_{1}) = \int_{\bar{\D}} g(\mathbf{x}) \d \mu
	=\lim_{\lambda\rightarrow0}\lambda u_{\lambda},
\end{align*}
where $u_{\lambda}$ is the solution of \eqref{PDE-ulambda} with $g(\mathbf{x}) = |y|\delta(x-a_{1})$. 
In practice, $g(\mathbf{x})$ can be approximated by
$\frac{|y|}{\sqrt{2\pi}\varepsilon_{0}}\exp\left(-\frac{(x-a_{1})^{2}}{2\varepsilon_{0}^{2}}\right)$ 
for a small $\varepsilon_{0} > 0$. 

Another direct approach is through Monte Carlo simulation. 
Let $\delta_{t} > 0$ be the time step, 
and take a sufficiently large final time $T = N\delta_{t}$. 
Define the discrete times $t_{n} = n\delta_{t}$ for \( n = 0, 1, \ldots, N \).  
We then construct a sequence $\{ (X_n, Y_n, Z_n), 0\leq n \leq N\}$. 
For each $n$, $(X_n, Y_n, Z_n)$ denotes an approximation of $(X(t_n), Y(t_n), Z(t_n))$.  
By using the property of Dirac delta function, we proceed with the following approximation
\begin{align*}
	\nu(a_{1}) &\approx \frac{1}{T}\int_{0}^{T}|Y(t)|\delta(X(t) - a_{1})\d t\\
	& = \frac{1}{T}\int_{0}^{T}|Y(t)| \sum_{\substack{\{j: X_{j} = a_{1}, Y_{j}\neq 0\}}}\frac{\delta(t-t_{j})}{|Y_{j}|}\d t\\
	& = \frac{1}{T}\sum_{j} \indic_{\{ X_{j} = a_{1}, Y_{j}\neq 0\}}.  
\end{align*}

We compare the results computed using the finite difference approach to PDEs \eqref{PDE-ulambda} 
with the results obtained by Monte Carlo simulations. 
The settings of parameters are listed in Appendix \ref{app-A04}. 
As shown in Fig. \ref{fig:example01}, the results based on Monte Carlo are close to the PDE results. 

\begin{figure}[htpb]
	\centering
	\captionsetup{width=14cm}
	\captionsetup{font= footnotesize}
	\includegraphics[width=5in]{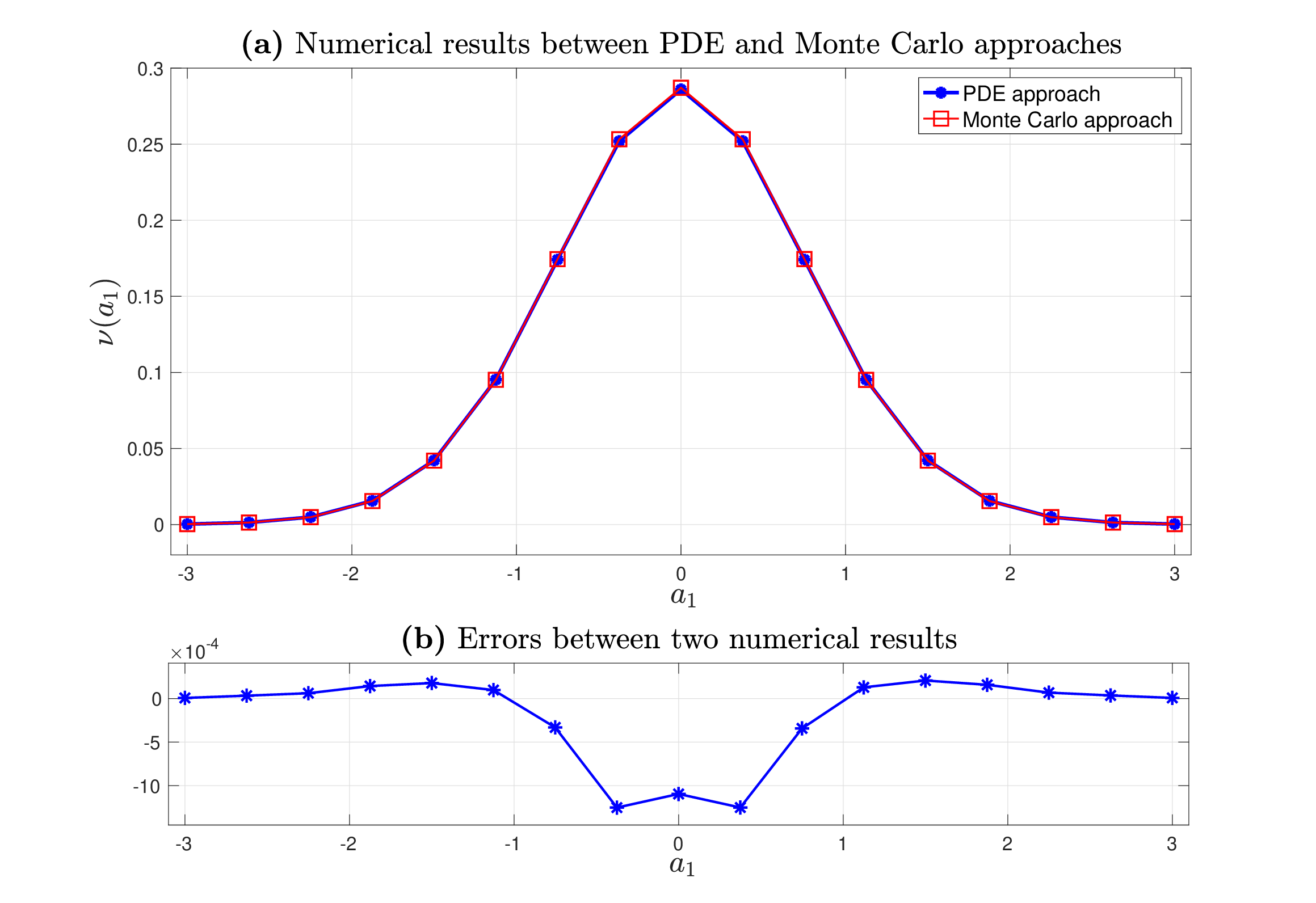}
	\caption{\footnotesize 	Frequency of threshold crossing per unit time interval for $a_{1} \in [-3,3]$.}
	\label{fig:example01}
\end{figure}

\subsection{Probability of the serviceability limit state}
Unlike elastic deformation, 
plastic deformation is permanent and cumulative. 
The deformation $\Delta(t) = X(t) - Z(t)$ can result either from a single large event or 
from the accumulation of many smaller ones, 
the latter being a direct consequence of the ratchet effect. 
In this subsection, we study the probability
$\mathbf{P}(a_{2}) \triangleq\lim\limits_{t\rightarrow \infty}\P[|X(t) - Z(t)|\leq a_{2}~|~ \mathbf{X}(0) = \mathbf{x}]$. 
By Theorem \ref{Thm-invariant-measure}, 
\begin{align*}
	\mathbf{P}(a_{2}) &= \lim_{t\rightarrow \infty}\E[\indic_{\{|X(t)-Z(t)| \leq a_{2}\}}(\mathbf{X}(t)) ~|~ \mathbf{X}(0) = \mathbf{x}]\\
	& = 	\lim_{\lambda\rightarrow0}\lambda u_{\lambda},
\end{align*}
where $u_{\lambda}$ is the solution of \eqref{PDE-ulambda} with
$g(\mathbf{x}) = \indic_{\{|x-z| \leq a_{2}\}}(\mathbf{x})$,
$a_{2}\geq0$.  
For the probabilistic numerical scheme, 
we use the Monte Carlo method, and the settings are the same as in Example 1. 
Then, $\mathbf{P}(a_{2})$ can be approximated by 
\begin{align*}
	\mathbf{P}(a_{2}) &\approx\frac{1}{N}\sum_{j = 1}^{N}\indic_{\{|X_{j}-Z_{j}| \leq a_{2}\}}.
\end{align*}

The numerical results obtained from PDE and Monte Carlo approaches are shown in Fig. \ref{fig:example02}. 
The settings of parameters can be found in Appendix \ref{app-A04}.

\begin{figure}[htpb]
	\centering
	\captionsetup{width=14cm}
	\captionsetup{font= footnotesize}
	\includegraphics[width=5in]{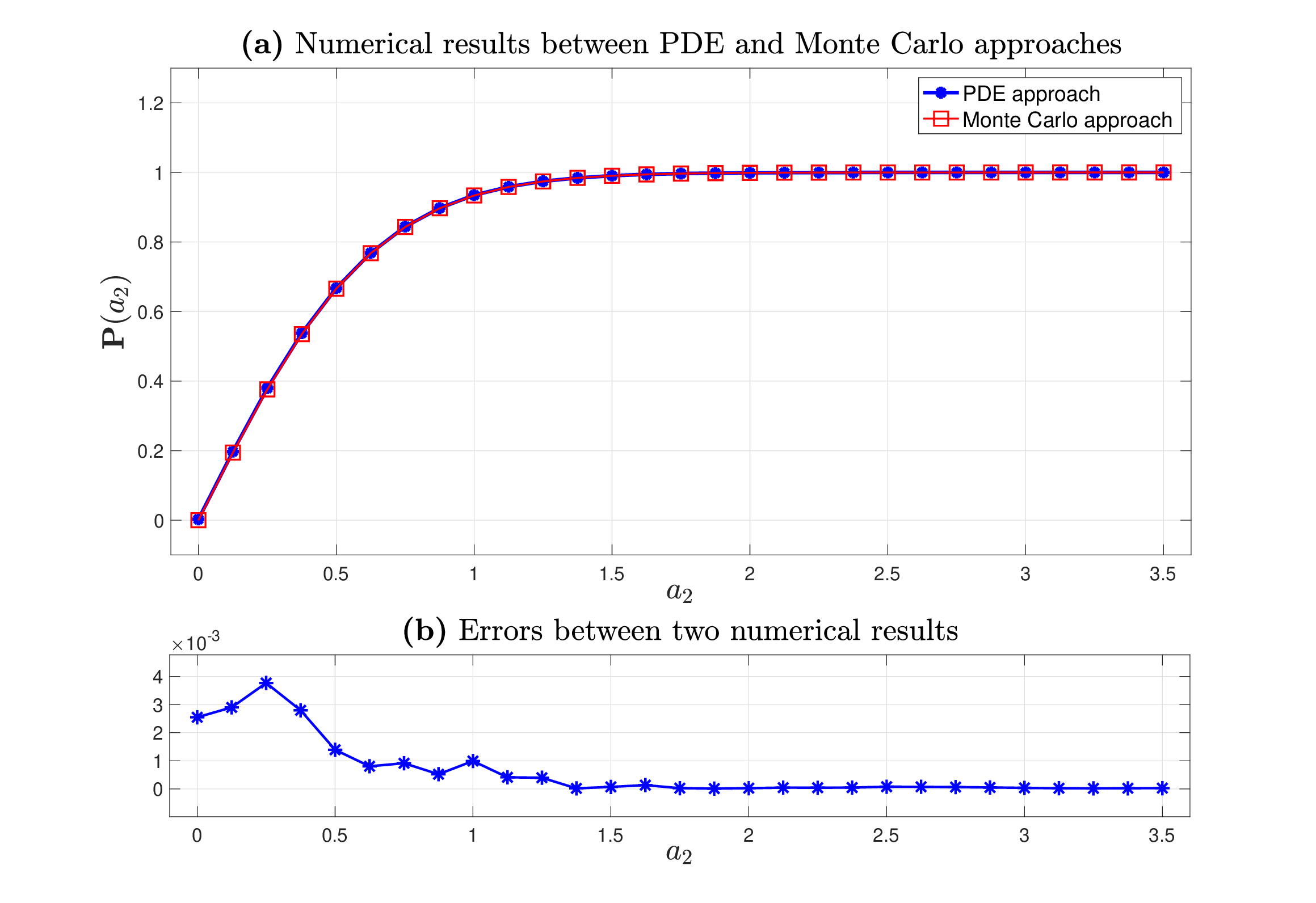}
	\caption{\footnotesize Probability of the serviceability in the limit state for $a_{2}\in[0,3.5]$.}
	\label{fig:example02}
\end{figure}

\subsection{Empirical convergence rate from the numerical case studies}
Since the exact solution is not always available, 
we employ numerical solutions computed on different meshes to estimate the empirical convergence rate 
of the PDE-based approach.   
To be precise, 
let $\mathcal{U}^{h}$ be the numerical solution on a uniform mesh with spacing $h$ and 
$\mathcal{U}^{\text{exact}}$ the exact solution, 
assume $p$ is the convergence rate for numerical schemes, then
\begin{align*}
	\|\mathcal{U}^{h} -  \mathcal{U}^{\text{exact}}\|_{\infty} \leq C h^{p},
\end{align*}
where $C$ is a constant independent of $h$. 
Moreover, assume that 
\begin{align*}
	\|\mathcal{U}^{h} -  \mathcal{U}^{\text{exact}}\|_{\infty} = C h^{p} + O(h^{p+\epsilon}),\quad \text{for~some~}\epsilon>0,
\end{align*}
then
\begin{align*}
	\frac{\|\mathcal{U}^{h} -  \mathcal{U}^{h/2}\|_{\infty}}{\|\mathcal{U}^{h/2} -  \mathcal{U}^{h/4}\|_{\infty}} 
	\approx 2^{p} + O(h^{\epsilon}). 
\end{align*}
With such a relation in mind, we test the convergence rate of the numerical scheme by considering
\begin{align*}
	p(h) \triangleq \log_{2}\left(\frac{\|\mathcal{U}^{h} -  \mathcal{U}^{h/2}\|_{\infty}}{\|\mathcal{U}^{h/2} -  \mathcal{U}^{h/4}\|_{\infty}}\right). 
\end{align*}
For more details, refer to \cite{bensoussan2021mathematical,roy2005review}.

In Tables \ref{Tab_01} and \ref{Tab_02}, we present a set of empirical estimations of $p(h)$ with respect to $x, y, z$ for Examples 1 and 2. 
The results show that the convergence rate is close to 2, 
which is consistent with the theoretical expectation of the second-order upwind difference scheme.

\begin{minipage}[t]{0.45\textwidth}
	\centering
	\captionof{table}{Computation $p(h)$ for Example 1, 
		with $\lambda = 10^{-3}$, $h_{x} = \frac{\lambda\bar{x}}{8}$, $h_{y} = \frac{\lambda\bar{y}}{8}$, 
		$h_{z} = \frac{\lambda b}{8}$, $\varepsilon_{0}=4\times 10^{-4}$, $a_{1} = 0$.}
	\begin{tabular}{c|c|c|c}
		\hline\hline
		\centering
		& $h = h_{x}$	& $h = h_{y}$& $h = h_{z}$ \\  \hline
		$\lambda\|u^{h} -  u^{h/2}\|_{\infty}$	&$0.0160$  & $0.0152$ & $0.0025$\\ \hline
		$\lambda\|u^{h/2} -  u^{h/4}\|_{\infty}$	&$0.0048$  & $0.0040$ & $0.0007$ \\ \hline
		$\lambda\|u^{h/4} -  u^{h/8}\|_{\infty}$	&$0.0011$  & $0.0010$ &  $0.0002$ \\ \hline\hline
		$p(h)$ & $1.7517$   & $1.9091$  & $1.9095$  \\ \hline
		$p(h/2)$ & $2.1561$   & $1.9856$  & $1.9185$  \\ \hline\hline
	\end{tabular}
	\label{Tab_01}
\end{minipage}
\hfill
\begin{minipage}[t]{0.45\textwidth}
	\centering
	\captionof{table}{Computation $p(h)$ for Example 2, 
		with $\lambda = 10^{-3}$, $h_{x} = \frac{\lambda\bar{x}}{8}$, $h_{y} = \frac{\lambda\bar{y}}{8}$, 
		$h_{z} = \frac{\lambda b}{8}$, $a_{2} = \frac{3}{8}$. }
	%\vspace{1.2em}
	\begin{tabular}{c|c|c|c}
		\hline\hline
		\centering
		& $h = h_{x}$	& $h = h_{y}$& $h = h_{z}$ \\  \hline
		$\lambda\|u^{h} -  u^{h/2}\|_{\infty}$	&$0.0973$  & $0.0321$ & $0.0411$\\ \hline
		$\lambda\|u^{h/2} -  u^{h/4}\|_{\infty}$	&$0.0414$  & $0.0100$ & $0.0122$ \\ \hline
		$\lambda\|u^{h/4} -  u^{h/8}\|_{\infty}$	&$0.0124$  & $0.0025$ &  $0.0031$ \\ \hline\hline
		$p(h)$ & $1.2341$   & $1.6886$  & $1.7481$  \\ \hline
		$p(h/2)$ & $1.7332$   & $1.9660$  & $1.9836$  \\ \hline\hline
	\end{tabular}
	\label{Tab_02}
\end{minipage}

\section{Conclusions}\label{sec:conclusions}
In this work, we designed a PDE-based numerical method for computing the invariant measure 
and steady-state statistics of a white-noise-driven bilinear hysteresis oscillator. 
The approach offers an alternative to Monte Carlo simulations. 
By establishing a theoretical extension of the EPPO model and 
providing an appropriate Lyapunov function, 
we demonstrated the existence of an invariant probability measure, 
although the uniqueness remains a challenging problem beyond the scope of this paper. 
Applications of the proposed framework, 
including the calculation of threshold crossing frequency and serviceability probability, 
illustrate its practical utility for analyzing high-dimensional nonlinear and nonsmooth dynamical systems. 
Numerical results confirm the efficiency and promise of our strategy 
for broadening the scope of alternatives to traditional Monte Carlo methods.

\appendix
\section{Proofs of theorems}\label{app-A01}

\renewcommand{\thecorollary}{A.\arabic{corollary}}
\renewcommand{\theequation}{A.\arabic{equation}}
\setcounter{figure}{0}
\renewcommand{\thefigure}{A-\arabic{figure}}
\renewcommand{\thetheorem}{A.\arabic{theorem}}
\renewcommand{\thelemma}{A.\arabic{lemma}}

Firstly, let $\mathbf{X}$ solve SVI \eqref{SVIBEPO}, 
we note that for a smooth enough function $F$, 
\begin{align}
	F(t, \mathbf{X}(t))
	=& ~F(0, \mathbf{X}(0))
	+ \int_0^t \partial_s F(s, \mathbf{X}(s)) \d s
	+ \int_0^t \partial_x F(s,\mathbf{X}(s)) \d X(s)
	+ \int_0^t \partial_y F(s,\mathbf{X}(s)) \d Y(s) \label{Ito-formula}\\
	&+ \frac{\sigma^2}{2}\int_0^t \partial^{2}_{y} F(s,\mathbf{X}(s)) \d s
	+ \underbrace{\int_0^t \partial_z F(s,\mathbf{X}(s)) \d Z(s),}
	_{=\int_0^t \partial_z F(s,\mathbf{X}(s)) Y(s) \mathrm{d} s 
		- \int_0^t \partial_z F(s,\mathbf{X}(s)) \mathrm{d} \Delta(s)}\notag
\end{align}
where $Z(t) = X(t) - \Delta(t)$ and 
$\d\Delta(t) = \indic_{\{Z(t) = b\}} \max(0, Y(t)) \d t + \indic_{\{Z(t) = -b\}} \min(0, Y(t)) \d t$. 

{\noindent \it \bf Proof of Theorem \ref{Prop01_BEPO}.}  
We prove that if $\psi$ is a solution of \eqref{SVI-eq01}, 
then $\psi$ satisfies \eqref{SVI-sol01}. 
For any $t,\tau \in [0,T]$,  
by It\^{o}'s formula (regard $t$ as a constant and refer to \eqref{Ito-formula}), 
\begin{align*}
	&\psi(t - \tau, \mathbf{X}(\tau)) \\
	&= \psi(t, \mathbf{X}(0))  
	- \int_{0}^{\tau} \partial_{s}\psi \d s 
	+  \int_0^\tau \partial_x \psi \d X(s)
	+ \int_0^\tau\partial_y \psi \d Y(s)
	+ \frac{\sigma^2}{2}\int_0^\tau\partial^{2}_{y} \psi \d s\\
	& ~~+ \int_0^\tau Y(s) \partial_z \psi \mathrm{d} s 
	- \int_0^\tau \partial_z \psi 
	\indic_{\{Z(s) = b\}} \max(0, Y(s)) \d s 
	-  \int_0^\tau \partial_z \psi \indic_{\{Z(s) = -b\}} \min(0, Y(s)) \d s. 
\end{align*}
Applying 
\begin{align*}
	&1 = \indic_{\{|Z(s)| < b\}} + \indic_{\{Z(s) = b\}}  + \indic_{\{Z(s) = -b\}},\\
	&Y(s) - \max(0, Y(s)) =  \min(0, Y(s)),
\end{align*}
and noting that $\psi$ solves  \eqref{SVI-eq01}, then
\begin{align*}
	\psi(t - \tau, \mathbf{X}(\tau)) = \psi(t, \mathbf{X}(0))  
	+ \text{mean-zero~term}.
\end{align*}
Since $\tau$ is arbitrary, let $\tau = t$, 
take the expectation conditional on $\mathbf{X}(0) = \mathbf{x}$, 
\begin{align*}
	\E[\psi(0 , \mathbf{X}(t)) |  \mathbf{X}(0) = \mathbf{x}]
	=\E[\psi(t, \mathbf{X}(0)) | \mathbf{X}(0) = \mathbf{x} ],
\end{align*}
that is 
$\E[g(\mathbf{X}(t)) | \mathbf{X}(0) = \mathbf{x}] = \psi(t, \mathbf{x})$. 	\qed

Especially, if the function $g$ in Theorem \ref{Prop01_BEPO} is replaced by an indicator function, 
we can obtain the following corollary. 
\begin{corollary}
	Let $\mathcal{B}$ be a Borel $\sigma$-field on $\bar{\mathcal{D}}$, 
	a function $u(s, \mathbf{x})$ given by
	\begin{align*}
		u(s, \mathbf{x}) 
		= \E[\indic_{\mathcal{B}}(\mathbf{X}(t)) | \mathbf{X}(s) = \mathbf{x} ] 
		= \mathbb{P}_{(s,\mathbf{x})}(\mathbf{X}(t)\in \mathcal{B}),
	\end{align*}
	with $0\leq s\leq t$ and $\mathbf{x}\in \bar{\D}$, is a solution of the equations
	\begin{align*}
		\begin{cases}
			\partial_{s} u  + \AA u = 0,~~\text{in~} [0,t)\times \D,\\
			\partial_{s} u + \BB_{\pm} u = 0,~~\text{in~}[0,t)\times \D^{\pm},\\
			u(t,\mathbf{x}) = \indic_{\mathcal{B}}(\mathbf{x}). 
		\end{cases}
	\end{align*}
\end{corollary}

{\noindent \it \bf Proof of Theorem \ref{prop-03}.}  
We consider the system in the elastic regime, for the plastic regime, the same result can be obtained. 
By It\^{o}'s formula, 
\begin{align*}
	\frac{\d }{\d t}\E[X^{2}(t)] &= 2\E[X(t)Y(t)], \\
	\frac{\d }{\d t}\E[Y^{2}(t)] &= 2\E[Y(t)\mathfrak{f}(X(t), Y(t)) -k(1-\alpha)Y(t)Z(t) -k\alpha X(t)Y(t)] + \sigma^{2},\\
	\frac{\d }{\d t}\E[X(t)Y(t)] &= \E[Y^{2}(t) + X(t)\mathfrak{f}(X(t), Y(t)) - k(1-\alpha)X(t)Z(t) -k\alpha X^{2}(t)]. 
\end{align*}
Recall that 
\begin{align*}
	V(x, y)  = \left(k\alpha +\frac{c_{0}d_{0}}{2} - c_{1}\right)x^{2} + y^{2} + c_{0} xy. 
\end{align*}
By using the assumptions \eqref{assum-01} and \eqref{assum-02}, 
and the Young's inequality, we have
\begin{align*}
	&\frac{\d }{\d t}\E[V(X(t), Y(t))] \\
	& \leq  \E[-c_{0} Y^{2}(t) - c_{0}(k\alpha - d_{1})X^{2}(t) - 2k(1-\alpha)Y(t)Z(t) - c_{0}k(1-\alpha) X(t)Z(t)] 
	+ \hat{C}\\
	&\leq -\frac{c_{0}}{2} \E[ (k\alpha - d_{1})X^{2}(t) + Y^{2}(t)] + C,
\end{align*}
where $\hat{C} = \sigma^{2} + c_{0}d_{2} + 2 c_{2}$, 
$C = \hat{C} + (kb(1-\alpha))^{2}\left(\frac{2}{c_{0}} + \frac{c_{0}}{2(k\alpha - d_{1})}\right)$. 
Moreover, since
\begin{align*}
	V(x, y) \leq \left(k\alpha +\frac{c_{0}d_{0}}{2} + \frac{c^{2}_{0}}{2} - c_{1}\right)x^{2} + \frac{3}{2}y^{2}, ~\forall x,y\in\R,
\end{align*}
then
\begin{align*}
	C_{1}V(x, y) \leq \frac{c_{0}}{2}[(k\alpha - d_{1})x^{2} + y^{2}], 
\end{align*}
where $C_{1} = \min\left(\frac{c_{0}}{3}, \frac{c_{0}(k\alpha - d_{1})}{2k\alpha +c_{0}d_{0} + c^{2}_{0} - 2c_{1}}\right)$. 
Therefore,
\begin{align*}
	\frac{\d }{\d t}\E[V(X(t), Y(t))] + C_{1}\E[V(X(t), Y(t))]  \leq C,
\end{align*}
which implies that 
\begin{align*}
	\E[V(X(t), Y(t))] &\leq V(X(0), Y(0))e^{-C_{1}t} + C(1-e^{-C_{1}t})/C_{1} \\
	&\leq V(X(0), Y(0))+ C/C_{1},~~~\forall t>0. 
\end{align*} \qed

Define the probability law $\mu_{_T}$ on $(\bar{\mathcal{D}}, \mathcal{B})$ associated with $T>0$ by
\begin{align*}
	\mu_{_T}(\phi) \triangleq \frac{1}{T}\int_{0}^{T}\E[\phi (\mathbf{X}(t))]\d t 
	=  \frac{1}{T}\int_{0}^{T} \mu(0) (P_{t}\phi) \d t, ~~\text{for}~\phi\in	C_{b}(\bar{\mathcal{D}}). 
\end{align*}
In fact, $\mu_{_T}$ corresponds to the ergodic mean. 
(For the definition of ergodic, refer to \cite{PK92}, p. 154.)

To prove Theorem \ref{Thm-01}, the following lemma is necessary. 
\begin{lemma}\label{lemma-02}
	The sequence $\{\mu_{_{T_{n}}}\}_{n\geq 1}$ is tight, for any $T_{n}\uparrow \infty$. 
\end{lemma}
 
\begin{proof}
	We need to prove that for any $\varepsilon>0$ and any $n\in\N$, 
	there exists a compact set $\mathcal{K}_{\varepsilon}\in\mathcal{B}$, 
	such that
	\begin{align*}
		\mu_{_{T_{n}}}(\indic_{\mathcal{K}_{\varepsilon}})\geq 1-\varepsilon. 
	\end{align*}
	Notice that 
	\begin{align*}
		\mu_{_{T_{n}}}(\indic_{\mathcal{K}_{\varepsilon}})
		= \frac{1}{T_{n}}\int_{0}^{T_{n}}\E[\indic_{\mathcal{K}_{\varepsilon}}(\mathbf{X}(t))]\d t 
		= \frac{1}{T_{n}}\int_{0}^{T_{n}}\P (\mathbf{X}(t) \in \mathcal{K}_{\varepsilon}) \d t. 
	\end{align*}
	Therefore,
	\begin{align*}
		1 - \mu_{_{T_{n}}}(\indic_{\mathcal{K}_{\varepsilon}})
		= \frac{1}{T_{n}}\int_{0}^{T_{n}}\P (\mathbf{X}(t) \notin \mathcal{K}_{\varepsilon}) \d t. 
	\end{align*}
	We then show that the set 
	\begin{align*}
		\mathcal{K}_{\delta} \triangleq 
		\left\{ (x,y,z) : \mathfrak{h} |x| + \Big|\frac{c_{0}}{2}x + y\Big| \leq \frac{1}{\delta}, |z| \leq b\right\} 
	\end{align*}
	is appropriate, where $\delta = \frac{1}{2}\sqrt{\frac{\varepsilon}{ \mathfrak{c}}}$, 
	$ \mathfrak{c} = V(X(0), Y(0)) + C/C_{1}$, and 
	$\mathfrak{h} = \left( k\alpha + c_{0}d_{0}/2 - c_{1} - c_{0}^{2}/4\right)^{1/2}$. 
	To see this, by applying Theorem \ref{prop-03} and Chebyshev's inequality, we have
	\begin{align*}
		\P (\mathbf{X}(t) \notin \mathcal{K}_{\varepsilon}) 
		& \leq 
		\P\left(\mathfrak{h} |X(t)| > \frac{1}{2\delta}\right) + \P\left(\Big|\frac{c_{0}}{2}X(t) + Y(t)\Big| > \frac{1}{2\delta}\right)\\
		& \leq 4\delta^{2}\E\left[\mathfrak{h}^{2} |X(t)|^{2} +  \Big|\frac{c_{0}}{2}X(t) + Y(t)\Big|^{2}\right]\\
		& \leq 4\delta^{2} \mathfrak{c} = \varepsilon,
	\end{align*}
	which implies the result. 
\end{proof}

{\noindent \it \bf Proof of Theorem \ref{Thm-01}.}  
For any $\mu(0)$, $\{\mu_{_{T_{n}}}\}_{n\geq 1}$ is tight according to Lemma \ref{lemma-02}, 
hence there exists a subsequence $\{\mu_{_{T_{n_{\rk}}}}\}_{\rk\geq 1}$ that converges weakly to a certain measure $\mu$, 
i.e., for any $\phi\in C_{b}(\bar{\mathcal{D}})$, 
\begin{align*}
	\mu_{_{T_{n_{\rk}}}}(\phi) \xrightarrow[\rk]{ } \mu (\phi). 
\end{align*}
Then $\mu$ is invariant, since
\begin{align*}
	\mu(P_{t} \phi) & = \lim_{\rk}\mu_{_{T_{n_{\rk}}}}(P_{t}\phi)\\
	& = \lim_{\rk}\frac{1}{T_{n_{\rk}}}\int_{0}^{T_{n_{\rk}}} \mu(0)(P_{s}P_{t}\phi)\d s\\
	& = \lim_{\rk}\frac{1}{T_{n_{\rk}}}\int_{0}^{T_{n_{\rk}}} \mu(0)(P_{s+t}\phi)\d s\\
	& = \lim_{\rk}\frac{1}{T_{n_{\rk}}}\int_{t}^{t+ T_{n_{\rk}}} \mu(0)(P_{s}\phi)\d s\\
	& = \lim_{\rk}\frac{1}{T_{n_{\rk}}}\int_{0}^{T_{n_{\rk}}} \mu(0)(P_{s}\phi)\d s\\
	& = \mu(\phi). 
\end{align*} \qed

{\noindent \it \bf Proof of Corollary \ref{coro02}.}  
By the definition of invariant measure (refer to \cite{BA98}, pp. 55-56),
we have  
\begin{align*}
	\int_{\bar{\mathcal{D}}} \Lambda(\varphi) \d \mu = 0,~~\forall \varphi\in C_{b}(\bar{\mathcal{D}}), 
\end{align*}
where $\Lambda$ is the infinitesimal generator of $P_{t}$. 
That is also
\begin{align*}
	\int_{\mathcal{D}} \AA\varphi(x,y,z) \d \mu 
	+ \int_{\mathcal{D}^+} \BB_{+}\varphi(x,y,b) \d \mu
	+ \int_{\mathcal{D}^-} \BB_{-}\varphi(x,y,- b) \d \mu
	= 0. 
\end{align*}\qed

{\noindent \it \bf Proof of Theorem \ref{Thm-invariant-measure}.}  
(i) Consider the Laplace transformation to the solution of \eqref{SVI-eq01} in $\mathcal{D}$, 
denoted as $u_{\lambda}$: 
\begin{align*}
	u_{\lambda}(\mathbf{x}) = \int_{0}^{\infty} e^{-\lambda t}\psi(t,\mathbf{x}) \d t, ~~\text{for}~\lambda>0.
\end{align*}
Integration by parts gives that 
\begin{align*}
	\lambda u_{\lambda}(\mathbf{x}) 
	& = -e^{-\lambda t}\psi(t,\mathbf{x})\Big|_{t=0}^{\infty} + 
	\int_{0}^{\infty} e^{-\lambda t}\partial_{t} \psi(t,\mathbf{x}) \d t\\
	& = \psi(0, \mathbf{x}) + \int_{0}^{\infty}e^{-\lambda t}\AA \psi(t,\mathbf{x}) \d t.
\end{align*}
Note that the operator $\AA$ is independent of $t$, we get
\begin{align*}
	\lambda u_{\lambda}(\mathbf{x}) 
	= g(\mathbf{x}) + \AA u_{\lambda}(\mathbf{x}). 
\end{align*}
Using the same way, on the boundaries $\mathcal{D}^\pm$, we have 
\begin{align*}
	\lambda u_{\lambda}(x,y,\pm b) 
	- \BB_{\pm} u_{\lambda}(x,y,\pm b)
	= g(x,y,\pm b). 
\end{align*}
The Laplace transformation implies the boundedness:
\begin{align*}
	|u_{\lambda }(\mathbf{x})| 
	\leq \sup_{[0,T]\times \bar{\mathcal{D}}} |\psi(t, \mathbf{x})| \int_{0}^{\infty}e^{-\lambda t}\d t
	\leq \frac{\| g \|_{\infty}}{\lambda}. 
\end{align*}

(ii) Thanks to the final value theorem for the Laplace transformation, 
\begin{align*}
	\lim_{\lambda \rightarrow 0}\lambda u_{\lambda}(\mathbf{x}) 
	= \lim_{t\rightarrow \infty} \psi(t, \mathbf{x}). 
\end{align*}
Hence,
\begin{align*}
	\lim_{\lambda \rightarrow 0}\lambda u_{\lambda}(\mathbf{x}) 
	= \lim_{t\rightarrow \infty} \E[g(\mathbf{X}(t)) | \mathbf{X}(0) = \mathbf{x}]
	= \int_{\bar{\mathcal{D}}}g\d \mu,~\forall \mathbf{x}\in\bar{\mathcal{D}}.  
\end{align*} \qed

\section{Algorithm for the invariant measure}\label{app-A02}
\setcounter{equation}{0}
\renewcommand{\theequation}{B.\arabic{equation}}
\setcounter{figure}{0}
\renewcommand{\thefigure}{B-\arabic{figure}}

We first classify the grid $\cG $ on $\mathcal{D}_{XY}$ according to 
the following label subsets, as illustrated in Fig. \ref{fig:discretization}. 
\begin{align*}
	\bm{K} &\triangleq \Big\{(\ri,\rj,\rk): \ri\in\{2,\ldots, I-1\}, \rj\in \{2,\ldots, J-1\}, \rk\in\{2,\ldots, K-1\} \Big\}, \\
	\bm{R}_{-} & \triangleq  \Big\{(\ri,\rj,\rk): \ri\in\{2,\ldots, I-1\}, \rj\in \{2,\ldots, J-1\}, \rk\in\{1\}\Big\},\\
	\bm{R}_{+} & \triangleq \Big\{(\ri,\rj,\rk): \ri\in\{2,\ldots, I-1\}, \rj\in \{2,\ldots, J-1\}, \rk\in\{K\}\Big\},\\
	\bm{B}_{-} & \triangleq  \Big\{(\ri,\rj,\rk): \ri\in\{1\}, \rj\in \{2,\ldots, J-1\}, \rk\in\{2,\ldots, K-1\}\Big\},\\
	\bm{B}_{+} & \triangleq  \Big\{(\ri,\rj,\rk): \ri\in\{I\}, \rj\in \{2,\ldots, J-1\}, \rk\in\{2,\ldots, K-1\}\Big\}, \\
	\hat{\bm{B}}_{-} & \triangleq  \Big\{(\ri,\rj,\rk): \ri\in\{1\}, \rj\in \{2,\ldots, J-1\}, \rk\in\{1, K\}\Big\},\\
	\hat{\bm{B}}_{+} & \triangleq  \Big\{(\ri,\rj,\rk): \ri\in\{I\}, \rj\in \{2,\ldots, J-1\}, \rk\in\{1, K\}\Big\},\\
	\bm{G} &\triangleq \Big\{(\ri,\rj,\rk): \ri\in\{1,2,\ldots,I\}, \rj\in \{1,J\}, \rk\in\{1,2,\ldots, K\}\Big\}. 
\end{align*}
The discretization systems in Section \ref{sec:num-method} can be rewritten as the 
following linear system 
\begin{align*}
	M_{\lambda}\bm{v} = \tilde{\bm{g}}, 
\end{align*}
where $M_{\lambda}$ is a sparse $IJK \times IJK$ matrix. 
Precisely, let $l(\ri,\rj,\rk) \triangleq \rk + (\rj -1)K + (\ri-1)JK$, 
and let $(\tilde{\bm{g}})_{\ri,\rj,\rk} = 0$, $\forall (\ri,\rj,\rk) \in \bm{G}$, 
otherwise, $(\tilde{\bm{g}})_{\ri,\rj,\rk} = \tilde{g}(\tilde{x}_{\ri}, \tilde{y}_{\rj}, \tilde{z}_{\rk})$. 
\begin{remark}
	Note that $l(\ri,\rj,\rk) \triangleq \rk + (\rj -1)K + (\ri-1)JK$ is the 
	mixed-radix representation of an integer in the interval $[1, IJK]$, 
	and this representation is unique, for more details, refer to \cite{ST67} pp. 41--42. 
\end{remark}

\subsection{Second-order upwind scheme}
Since the second-order upwind scheme involves two adjacent upstream (or downstream) nodes, 
special treatment is required for finite difference nodes near the boundary to 
avoid referencing values outside the computational domain. 
Following the approach in \cite{vanka1987second}, 
this study reverts to the first-order upwind scheme for such boundary-adjacent points.
In this case,  
the elements of matrix $M_{\lambda}$ can be elaborated as follows. 

\vspace{0.5em}
\begin{enumerate}[label = (\arabic*)]
	\setlength{\itemsep}{4pt}
	%(1)
	\item For $(\ri,\rj,\rk) \in \bm{K}$, 
	\begin{align}
		M_{\lambda}(l(\ri,\rj,\rk), l(\ri -2,\rj,\rk)) 
		&= -\frac{\min(0,\tilde{y}_{\rj}/\lambda)}{2\delta_{\tilde{x}}}, \quad \text{if~}2<\ri < I-1, \label{2order-i1}\\
		M_{\lambda}(l(\ri,\rj,\rk), l(\ri +2,\rj,\rk)) 
		&= \frac{\max(0,\tilde{y}_{\rj}/\lambda)}{2\delta_{\tilde{x}}}, \quad \text{if~}2<\ri < I-1, \label{2order-i2}\\
		M_{\lambda}(l(\ri,\rj,\rk), l(\ri -1,\rj,\rk)) 
		&= (1 + \indic_{\{2< \ri < I-1\}})\frac{\min(0,\tilde{y}_{\rj}/\lambda)}{\delta_{\tilde{x}}}, \label{2order-i3}\\
		M_{\lambda}(l(\ri,\rj,\rk), l(\ri +1,\rj,\rk)) 
		&= -(1 + \indic_{\{2< \ri < I-1\}})\frac{\max(0,\tilde{y}_{\rj}/\lambda)}{\delta_{\tilde{x}}}, \label{2order-i4}\\
		%-------------------------------------------------------------------
		M_{\lambda}(l(\ri,\rj,\rk), l(\ri ,\rj-2,\rk)) 
		&= -\frac{\min(0,\tilde{\beta}_{\ri,\rj,\rk})}{2\delta_{\tilde{y}}},\quad \text{if~}2<\rj < J-1, \label{2order-j1}\\
		M_{\lambda}(l(\ri,\rj,\rk), l(\ri ,\rj+2,\rk)) 
		&= \frac{\max(0, \tilde{\beta}_{\ri,\rj,\rk})}{2\delta_{\tilde{y}}}, \quad \text{if~}2<\rj < J-1, \label{2order-j2}\\
		M_{\lambda}(l(\ri,\rj,\rk), l(\ri ,\rj-1,\rk)) 
		&= - \frac{\lambda\sigma^{2}}{2\delta_{\tilde{y}}^{2}} + (1 + \indic_{\{2< \rj < J-1\}})\frac{\min(0,\tilde{\beta}_{\ri,\rj,\rk})}{\delta_{\tilde{y}}}, \label{2order-j3}\\
		M_{\lambda}(l(\ri,\rj,\rk), l(\ri ,\rj+1,\rk)) 
		&= - \frac{\lambda\sigma^{2}}{2\delta_{\tilde{y}}^{2}} - (1 + \indic_{\{2< \rj < J-1\}})\frac{\max(0,\tilde{\beta}_{\ri,\rj,\rk})}{\delta_{\tilde{y}}},\label{2order-j4}\\
		%-------------------------------------------------------------------
		M_{\lambda}(l(\ri,\rj,\rk), l(\ri,\rj,\rk -2)) 
		&= -\frac{ \min(0,\tilde{y}_{\rj}/\lambda)}{2\delta_{\tilde{z}}}, \quad \text{if~}2<\rk < K-1, \label{2order-k1}\\
		M_{\lambda}(l(\ri,\rj,\rk), l(\ri,\rj,\rk + 2)) 
		&= \frac{ \max(0,\tilde{y}_{\rj}/\lambda)}{2\delta_{\tilde{z}}}, \quad \text{if~}2<\rk < K-1, \label{2order-k2}\\
		M_{\lambda}(l(\ri,\rj,\rk), l(\ri,\rj,\rk -1)) 
		&= (1 + \indic_{\{2< \rk < K-1\}})\frac{ \min(0,\tilde{y}_{\rj}/\lambda)}{\delta_{\tilde{z}}}, \label{2order-k3}\\
	    M_{\lambda}(l(\ri,\rj,\rk), l(\ri,\rj,\rk + 1)) 
		&= -(1 + \indic_{\{2< \rk < K-1\}})\frac{ \max(0,\tilde{y}_{\rj}/\lambda)}{\delta_{\tilde{z}}}, \label{2order-k4}\\
		%-------------------------------------------------------------------
		M_{\lambda}(l(\ri,\rj,\rk), l(\ri,\rj,\rk)) 
		&= 1 + \frac{\lambda\sigma^{2}}{\delta_{\tilde{y}}^{2}} 
		+ (2 +\indic_{\{2< \rj < J-1\}})\frac{|\tilde{\beta}_{\ri,\rj,\rk}|}{2\delta_{\tilde{y}}}\\
		&~~~~+ (2 +\indic_{\{2< \ri < I-1\}})\frac{|\tilde{y}_{\rj}/\lambda|}{2\delta_{\tilde{x}}} 
		+ (2 +\indic_{\{2< \rk < K-1\}})\frac{|\tilde{y}_{\rj}/\lambda|}{2\delta_{\tilde{z}}}. \notag
	\end{align}
	%(2)
	\item  For $(\ri,\rj,\rk) \in \bm{R}_{-}$, 
	elements \eqref{2order-i1}--\eqref{2order-j4} and \eqref{2order-k2} match the form given in case (1), 
	but 
	\begin{align*}
		M_{\lambda}(l(\ri,\rj,\rk), l(\ri,\rj,\rk+1)) 
		&= -\frac{2\max(0, \tilde{y}_{\rj}/\lambda)}{\delta_{\tilde{z}}}, \\
		M_{\lambda}(l(\ri,\rj,\rk), l(\ri,\rj,\rk)) 
		&= 1 + \frac{\lambda\sigma^{2}}{\delta_{\tilde{y}}^{2}} 
		+ (2 +\indic_{\{2< \rj < J-1\}})\frac{|\tilde{\beta}_{\ri,\rj,\rk}|}{2\delta_{\tilde{y}}}
		+ (2 +\indic_{\{2< \ri < I-1\}})\frac{|\tilde{y}_{\rj}/\lambda|}{2\delta_{\tilde{x}}} \\
		&~~~~+ \frac{3\max(0, \tilde{y}_{\rj}/\lambda)}{2\delta_{\tilde{z}}}.
	\end{align*}
	%(3)
	\item For $(\ri,\rj,\rk) \in \bm{R}_{+}$, 
	elements \eqref{2order-i1}--\eqref{2order-k1} match the form given in case (1), 
	but 
	\begin{align*}
		M_{\lambda}(l(\ri,\rj,\rk), l(\ri,\rj,\rk-1)) 
		&= \frac{2\min(0, \tilde{y}_{\rj}/\lambda)}{\delta_{\tilde{z}}}, \\
		M_{\lambda}(l(\ri,\rj,\rk), l(\ri,\rj,\rk)) 
		&= 1 + \frac{\lambda\sigma^{2}}{\delta_{\tilde{y}}^{2}} 
		+ (2 +\indic_{\{2< \rj < J-1\}})\frac{|\tilde{\beta}_{\ri,\rj,\rk}|}{2\delta_{\tilde{y}}}
		+ (2 +\indic_{\{2< \ri < I-1\}})\frac{|\tilde{y}_{\rj}/\lambda|}{2\delta_{\tilde{x}}} \\
		&~~~~ - \frac{3\min(0, \tilde{y}_{\rj}/\lambda)}{2\delta_{\tilde{z}}}. 
	\end{align*}
	%(4)
	\item For $(\ri,\rj,\rk) \in \bm{B}_{-}$, 
	elements \eqref{2order-i2} and \eqref{2order-j1}--\eqref{2order-k4} remain unchanged from case (1), 
	but 
	\begin{align*}
		M_{\lambda}(l(\ri,\rj,\rk), l(\ri+1,\rj,\rk)) 
		&= -\frac{2\max(0, \tilde{y}_{\rj}/\lambda)}{\delta_{\tilde{x}}}, \\
		M_{\lambda}(l(\ri,\rj,\rk), l(\ri,\rj,\rk)) 
		&= 1 + \frac{\lambda\sigma^{2}}{\delta_{\tilde{y}}^{2}} 
		+ (2 +\indic_{\{2< \rj < J-1\}})\frac{|\tilde{\beta}_{\ri,\rj,\rk}|}{2\delta_{\tilde{y}}}
		+ \frac{3\max(0, \tilde{y}_{\rj}/\lambda)}{2\delta_{\tilde{x}}} \\
		&~~~~+ (2 +\indic_{\{2< \rk < K-1\}})\frac{|\tilde{y}_{\rj}/\lambda|}{2\delta_{\tilde{z}}}. 
	\end{align*}
	%(5)
	\item For $(\ri,\rj,\rk) \in \bm{B}_{+}$, 
	elements \eqref{2order-i1} and \eqref{2order-j1}--\eqref{2order-k4} remain unchanged from case (1), 
	but 
	\begin{align*}
		M_{\lambda}(l(\ri,\rj,\rk), l(\ri-1,\rj,\rk)) 
		&= \frac{2\min(0, \tilde{y}_{\rj}/\lambda)}{\delta_{\tilde{x}}}, \\
		M_{\lambda}(l(\ri,\rj,\rk), l(\ri,\rj,\rk)) 
		&=1 + \frac{\lambda\sigma^{2}}{\delta_{\tilde{y}}^{2}} 
		+ (2 +\indic_{\{2< \rj < J-1\}})\frac{|\tilde{\beta}_{\ri,\rj,\rk}|}{2\delta_{\tilde{y}}}
		- \frac{3\min(0, \tilde{y}_{\rj}/\lambda)}{2\delta_{\tilde{x}}} \\
		&~~~~+ (2 +\indic_{\{2< \rk < K-1\}})\frac{|\tilde{y}_{\rj}/\lambda|}{2\delta_{\tilde{z}}}. 
	\end{align*}
	%(6)
	\item For $(\ri,\rj,\rk) \in \{(\ri,\rj,\rk)\in \hat{\bm{B}}_{-}: \rk = 1\}$, 
	use case (1) for \eqref{2order-j1}--\eqref{2order-j4}, 
	case (2) for \eqref{2order-k2} and \eqref{2order-k4}, 
	and case (4) for \eqref{2order-i2} and \eqref{2order-i4}, 
	but 
	\begin{align*}
		M_{\lambda}(l(\ri,\rj,\rk), l(\ri,\rj,\rk)) 
		=1 + \frac{ \lambda\sigma^{2}}{\delta_{\tilde{y}}^{2}} 
		+ (2 +\indic_{\{2< \rj < J-1\}})\frac{|\tilde{\beta}_{\ri,\rj,\rk}|}{2\delta_{\tilde{y}}}
		+ \frac{3\max(0, \tilde{y}_{\rj}/\lambda)}{2\delta_{\tilde{x}}} 
		+ \frac{3\max(0, \tilde{y}_{\rj}/\lambda)}{2\delta_{\tilde{z}}}.
	\end{align*}
	%(7)
	\item For $(\ri,\rj,\rk) \in \{(\ri,\rj,\rk)\in \hat{\bm{B}}_{+}: \rk = 1\}$, 
	use case (1) for \eqref{2order-j1}--\eqref{2order-j4}, 
	case (2) for \eqref{2order-k2} and \eqref{2order-k4}, 
	and case (5) for \eqref{2order-i1} and \eqref{2order-i3}, 
	but 
	\begin{align*}
		M_{\lambda}(l(\ri,\rj,\rk), l(\ri,\rj,\rk)) 
		= 1 + \frac{\lambda\sigma^{2}}{\delta_{\tilde{y}}^{2}} 
		+ (2 +\indic_{\{2< \rj < J-1\}})\frac{|\tilde{\beta}_{\ri,\rj,\rk}|}{2\delta_{\tilde{y}}}
		- \frac{3\min(0, \tilde{y}_{\rj}/\lambda)}{2\delta_{\tilde{x}}} 
		+ \frac{3\max(0, \tilde{y}_{\rj}/\lambda)}{2\delta_{\tilde{z}}}.  
	\end{align*}
	%(8)
	\item For $(\ri,\rj,\rk) \in \{(\ri,\rj,\rk)\in \hat{\bm{B}}_{-}: \rk = K\}$, 
	use case (1) for \eqref{2order-j1}--\eqref{2order-j4}, 
	case (3) for \eqref{2order-k1} and \eqref{2order-k3}, 
	case (4) for \eqref{2order-i2} and \eqref{2order-i4}, 
	but 
	\begin{align*}
		M_{\lambda}(l(\ri,\rj,\rk), l(\ri,\rj,\rk)) 
		= 1 + \frac{ \lambda\sigma^{2}}{\delta_{\tilde{y}}^{2}} 
		+ (2 +\indic_{\{2< \rj < J-1\}})\frac{|\tilde{\beta}_{\ri,\rj,\rk}|}{2\delta_{\tilde{y}}}
		+ \frac{3\max(0, \tilde{y}_{\rj}/\lambda)}{2\delta_{\tilde{x}}} 
		- \frac{3\min(0, \tilde{y}_{\rj}/\lambda)}{2\delta_{\tilde{z}}}.  
	\end{align*}
	%(9)
	\item For $(\ri,\rj,\rk) \in \{(\ri,\rj,\rk)\in \hat{\bm{B}}_{+}: \rk = K\}$, 
	use case (1) for \eqref{2order-j1}--\eqref{2order-j4}, 
	case (3) for \eqref{2order-k1} and \eqref{2order-k3}, 
	and case (5) for \eqref{2order-i1} and \eqref{2order-i3},  
	but
	\begin{align*}
		M_{\lambda}(l(\ri,\rj,\rk), l(\ri,\rj,\rk)) 
		= 1 + \frac{ \lambda\sigma^{2}}{\delta_{\tilde{y}}^{2}} 
		+ (2 +\indic_{\{2< \rj < J-1\}})\frac{|\tilde{\beta}_{\ri,\rj,\rk}|}{2\delta_{\tilde{y}}}
		- \frac{3\min(0, \tilde{y}_{\rj}/\lambda)}{2\delta_{\tilde{x}}} 
		- \frac{3\min(0, \tilde{y}_{\rj}/\lambda)}{2\delta_{\tilde{z}}}.  
	\end{align*}
	%(10)
	\item For $(\ri,\rj,\rk) \in \bm{G}$, 
	\begin{align*}
		&M_{\lambda}(l(\ri,\rj,\rk), l(\ri,\rj,\rk)) = -\frac{1}{\delta_{\tilde{y}}}, \quad 
		M_{\lambda}(l(\ri,\rj,\rk), l(\ri,\rj+1,\rk)) = \frac{1}{\delta_{\tilde{y}}}, \quad \text{if~} \rj = 1,\\ 
		&M_{\lambda}(l(\ri,\rj,\rk), l(\ri,\rj-1,\rk)) = -\frac{1}{\delta_{\tilde{y}}}, \quad 
		M_{\lambda}(l(\ri,\rj,\rk), l(\ri,\rj,\rk)) = \frac{1}{\delta_{\tilde{y}}}, \quad \text{if~} \rj = J.
	\end{align*}
\end{enumerate}

\section{Rice's formula}\label{app-A03}
\setcounter{equation}{0}
\renewcommand{\theequation}{C.\arabic{equation}}
\renewcommand{\thedefinition}{C.\arabic{definition}}
\setcounter{figure}{0}
\renewcommand{\thefigure}{C-\arabic{figure}}

Before introduce the Rice's formula, we need the following definition. 
\begin{definition}[Upcrossing and downcrossing \cite{L66}]
	For an almost sure (a.s.) continuous process $\zeta = \{\zeta(t), t\geq 0\}$, 
	we say it has a upcrossing of $a\in\R$ at $t_{0}>0$, 
	if for some $\varepsilon>0$, $\zeta(t)\leq a$ in $(t_{0}-\varepsilon, t_{0})$ and 
	$\zeta(t) \geq a$ in $(t_{0}, t_{0}+\varepsilon)$. 
	The downcrossing of the level $a$ is similarly defined by reversing inequalities. 
	
	For the a.s. differentiable process $\zeta$, it has a upcrossing (downcrossing) of $a$ at $t_{0}>0$, 
	if $\dot{\zeta}(t_{0}) > 0$ ($\dot{\zeta}(t_{0}) < 0$).  
\end{definition}

Next, we introduce the Rice's formula. 
Let $\eta = \{\eta(t), t\geq 0\}$ be a continuous and strictly stationary process. 
The distribution of $\eta(0)$, denoted by $F(x) = \P(\eta(0)\leq x)$, 
is assumed to be a continuous function. 
For a fixed $a\in \R$, let $N^{+}(a)$ be the number of upcrossing of level $a$ by $\eta$ on $[0,1]$, 
the intensity of upcrossing can be evaluated using the relation 
\begin{align*}
	\nu^{+}(a) \triangleq \E[N^{+}(a)] = \lim_{q\downarrow 0}\frac{1}{q}\P(\eta(0) < a < \eta(q)).
\end{align*}
The proof can be found in \cite{LLR83} Lemma 7.2.2. 
If the process $\eta$ is Gaussian and the derivative $\dot{\eta}(t)$ exists in a quadratic mean, 
then Rice's formula states that 
\begin{align*}
	\nu^{+}(a) = \int_{0}^{\infty} y p(a,y)\d y,
\end{align*}
where $p(a,y)$ is the joint density of $\eta(0)$ and $\dot{\eta}(0)$ (refer to \cite{LLR83} Theorem 7.2.4). 
A more general form of Rice's formula for the intensity of crossing a given level $a$:
\begin{align*}
	\nu(a) = \int_{-\infty}^{\infty} |y| p(a,y)\d y,
\end{align*}
where $\nu(a) = \nu^{+}(a) + \nu^{-}(a)$, and $ \nu^{-}(a)$ is the downcrossing intensity.

\section{Parameter settings for two numerical examples}\label{app-A04}
\setcounter{equation}{0}
\renewcommand{\theequation}{D.\arabic{equation}}
\setcounter{figure}{0}
\renewcommand{\thefigure}{D-\arabic{figure}}

Here, 
we introduce the parameter settings in the two examples in Section \ref{sec:num-results}.
For all numerical experiments, 
we utilized a MacBook Pro equipped with an Apple M4 Pro processor (14 -- core CPU configuration) 
and 24 GB of unified memory. 
Computations were executed in MATLAB R2024b. 

\vspace{0.5em}
\noindent{\bf{PDE approach} (by applying the second-order upwind scheme)}
\vspace{0.5em}

To obtain solutions of the two PDEs in Section \ref{sec:num-results}, 
we numerically solve the linear systems from Section \ref{sec:num-method}. 
In particular, we solve the linear systems by using the generalized minimal residuals (GMRES) 
algorithm implemented in MATLAB's {\bf{gmres}} function, 
and precondition each system with incomplete LU factorization. 

{\renewcommand
	\arraystretch{1.2}
	\tabcolsep = 3em
	\begin{center}
		\begin{tabular}{c|cc}
			\hline\hline
			\centering
			{\bf{Parameters}} 		& {\bf{Example 1}}     	& {\bf Example 2} \\ 
			\hline
			$a_{1}$	 or $a_{2}$  &$[-3,3]$  				  	& $[0,3.5]$ \\
			$\lambda$  				   &$10^{-3}$			      	&  $10^{-3}$ \\
			$\bar{x}$    				  &$3.5$						 	&  $3.5$	 \\
			$\bar{y}$    				  &$3.5$						 	&  $3.5$ \\
			$b$              				  &$1$							   	   &  $1$ \\
			$\mathfrak{f}(x,y)$ &$-y$								&  $-y$ \\
			$k$          					  &$1$							   	   &  $1$ \\
			$\alpha$					 &$0.5$								&  $0.5$ \\
			$\sigma$					&$1$							 	 &  $1$ \\
			$\delta_{\tilde{x}}$			  &$5.47\times 10^{-5}$					 	 &  $5.47\times 10^{-5}$ \\
			$\delta_{\tilde{y}}$			  &$5.47\times 10^{-5}$				         &  $5.47\times 10^{-5}$ \\
			$\delta_{\tilde{z}}$			  &$1.56\times 10^{-5}$					     &  $1.56\times 10^{-5}$ \\
			$IJK$						   &$2.15\times 10^{6}$	& $2.15\times 10^{6}$ \\ \hline\hline
		\end{tabular}
\end{center}}

\vspace{0.5em}
In Example  1, the Dirac delta function $\delta(x-a_{1})$ is approximated by 
$\frac{1}{\sqrt{2\pi}\varepsilon_{0}}\exp\left(-\frac{(x-a_{1})^{2}}{2\varepsilon_{0}^{2}}\right)$, 
with $\varepsilon_{0} = \lambda\bar{x}/2^6$. 

\vspace{1em}
\noindent{\bf{Monte Carlo approach}}
\vspace{1em}

{\renewcommand
	\arraystretch{1.2}
	\tabcolsep = 3em
	\begin{center}
		\begin{tabular}{c|cc}
			\hline\hline
			\centering 
			{\bf{Parameters}}  &{\bf{Example 1}}  & {\bf Example 2} \\ 
			\hline
			$a_{1}$	 or $a_{2}$  &$[-3,3]$  				  	& $[0,3.5]$ \\
			$\mathbf{x}$  &$(0,0,0)$  				  	& $(0,0,0)$ \\
			$b$              				  &$1$							   	   &  $1$ \\
			$\mathfrak{f}(x,y)$ &$-y$								&  $-y$ \\
			$k$          					  &$1$							   	   &  $1$ \\
			$\alpha$					 &$0.5$								&  $0.5$ \\
			$\sigma$					&$1$							 	 &  $1$ \\
			$\delta_{t} $						 &$10^{-3}$             &$10^{-3}$          \\
			$T$							 &$10^{9}$                &$10^{9}$           \\ \hline\hline
		\end{tabular}
\end{center}}

\bibliographystyle{unsrt}
\bibliography{Ref}   % name your BibTeX data base

@book{PK92,
	title={Numerical Solution of Stochastic Differential Equations},
	author={Kloeden, P. E. and Platen, E.},
	year={1992},
	publisher={Berlin: Springer-Verlag}
}

@book{BA98,
	title={Diffusions and Elliptic Operators},
	author={Bass, R. F.},
	year={1998},
	publisher={New York: Springer-Verlag}
}

@article{BJFMM12,
	title={Asymptotic analysis of stochastic variational inequalities modeling an elasto-plastic problem with vanishing jumps},
	author={Bensoussan, A. and Jasso-Fuentes, H. and Menozzi, S. and Mertz, L.},
	journal={Asymptotic Analysis},
	volume={80},
	number={1-2},
	pages={171--187},
	year={2012},
	publisher={Sage Publications Sage UK: London, England}
}

@article{BMPT09,
	title={An ultra weak finite element method as an alternative to a {M}onte {C}arlo method for an elasto-plastic problem with noise},
	author={Bensoussan, A. and Mertz, L. and Pironneau, O. and Turi, J.},
	journal={SIAM Journal on Numerical Analysis},
	volume={47},
	number={5},
	pages={3374--3396},
	year={2009},
	publisher={SIAM}
}

@article{BM12,
	title={An analytic approach to the ergodic theory of a stochastic variational inequality},
	author={Bensoussan, A. and Mertz, L.},
	journal={Comptes Rendus  Mathematique},
	volume={350},
	number={7-8},
	pages={365--370},
	year={2012}
}

@article{BFMY15,
	title={An analytical approach for the growth rate of the variance of the deformation related to an elasto-plastic oscillator excited by a white noise},
	author={Bensoussan, A. and Feau, C. and Mertz, L. and Yam, S. C. P.},
	journal={Applied Mathematics Research eXpress},
	volume={2015},
	number={1},
	pages={99--128},
	year={2015},
	publisher={Oxford University Press}
}

@article{BMY16,
	title={Nonlocal boundary value problems of a stochastic variational inequality modeling an elasto-plastic oscillator excited by a filtered noise},
	author={Bensoussan, A. and Mertz, L. and Yam, S. C. P.},
	journal={SIAM Journal on Mathematical Analysis},
	volume={48},
	number={4},
	pages={2783--2805},
	year={2016},
	publisher={SIAM}
}

@article{BT08,
	title={Degenerate {D}irichlet problems related to the invariant measure of elasto-plastic oscillators},
	author={Bensoussan, A. and Turi, J.},
	journal={Applied Mathematics and Optimization},
	volume={58},
	number={1},
	pages={1--27},
	year={2008},
	publisher={Springer}
}

@incollection{BT10,
	title={On a class of partial differential equations with nonlocal {D}irichlet boundary conditions},
	author={Bensoussan, A. and Turi, J.},
	booktitle={Applied and Numerical Partial Differential Equations: Scientific Computing in Simulation, Optimization and Control in a Multidisciplinary Context},
	pages={9--23},
	year={2009},
	publisher={Dordrecht: Springer-Verlag}
}

@article{C60,
	title={Random excitation of a system with bilinear hysteresis},
	author={Caughey, T. K.},
	journal={Journal of Applied Mechanics},
	volume={27},
	number={4},
	pages={649--652},
	year={1960},
	publisher={American Society of Mechanical Engineers Digital Collection}
}

@article{FM12,
	title={An empirical study on plastic deformations of an elasto-plastic problem with noise},
	author={Mertz, L. and Feau, C.},
	journal={Probabilistic Engineering Mechanics},
	volume={30},
	pages={60--69},
	year={2012},
	publisher={Elsevier}
}

@article{R78,
	title={The response of an oscillator with bilinear hysteresis to stationary random excitation},
	author = {Roberts, J. B.},
	journal = {Journal of Applied Mechanics},
	volume = {45},
	number = {4},
	pages = {923--928},
	year = {1978}
}

@book{PR14,
	title={Stochastic Differential Equations, Backward SDEs, Partial Differential Equations},
	author={Pardoux, E. and R\u{a}{\c{s}}canu, A.},
	volume={69},
	year={2014},
	publisher={Cham: Springer-Verlag}
}

@article{S79,
	title={Hysteretic structural vibrations under random load},
	author={Spanos, P-T. D.},
	journal={The Journal of the Acoustical Society of America},
	volume={65},
	number={2},
	pages={404--410},
	year={1979},
	publisher={Acoustical Society of America}
}

@article{YH87,
	title={Modeling and response of bilinear hysteretic systems},
	author={Yar, M. and Hammond, J. K.},
	journal={Journal of Engineering Mechanics},
	volume={113},
	number={7},
	pages={1000--1013},
	year={1987},
	publisher={American Society of Civil Engineers}
}

@article{LM16,
	title={Asymptotic formulae for the risk of failure related to an elasto-plastic problem with noise},
	author={Feau, C. and Lauri{\`e}re, M. and Mertz, L.},
	journal={Asymptotic Analysis},
	volume={106},
	number={1},
	pages={47--60},
	year={2017},
	publisher={Sage Publications Sage UK: London, England}
}

@article{NM96,
	title={Stationary and non-stationary random vibration of oscillators with bilinear hysteresis},
	author={Naess, A. and Moe, V.},
	journal={International Journal of Non-linear Mechanics},
	volume={31},
	number={5},
	pages={553--562},
	year={1996},
	publisher={Elsevier}
}

@article{MSW19,
	title={A {F}eynman--{K}ac formula approach for computing expectations and threshold crossing probabilities of non-smooth stochastic dynamical systems},
	author={Mertz, L. and Stadler, G. and Wylie, J.},
	journal={Physica D},
	volume={397},
	pages={25--38},
	year={2019},
	publisher={Elsevier}
}

@article{JFMY14,
	title={Approximate solutions of a stochastic variational inequality modeling an elasto-plastic problem with noise},
	author={Jasso-Fuentes, H. and Mertz, L. and Yam, S. C. P.},
	journal={Applied Mathematics Research Express},
	volume={2014},
	number={1},
	pages={52--73},
	year={2014},
	publisher={Oxford University Press}
}

@article{F08,
	title={Probabilistic response of an elastic perfectly plastic oscillator under {G}aussian white noise},
	author={Feau, C.},
	journal={Probabilistic Engineering Mechanics},
	volume={23},
	number={1},
	pages={36--44},
	year={2008},
	publisher={Elsevier}
}

@article{L66,
	title={On crossings of levels and curves by a wide class of stochastic processes},
	author={Leadbetter, M. R.},
	journal={The Annals of Mathematical Statistics},
	volume={37},
	number={6},
	year={1966},
	pages={260--267},
	publisher={JSTOR}
}

@book{LLR83,
	title={Extremes and Related Properties of Random Sequences and Processes},
	author={Leadbetter, M. R. and Lindgren, G. and Rootz{\'e}n, H.},
	year={1983},
	publisher={New York: Springer-Verlag}
}

@Book{ST67,
	title={Residue Arithmetic and Its Applications to Computer Technology},
	author={Szabo, N. S. and Tanaka, R. I.},
	publisher={New York: McGraw-Hill},
	year={1967}
}

@article{vanka1987second,
	title={Second-order upwind differencing in a recirculating flow},
	author={Vanka, S. P.},
	journal={AIAA Journal},
	volume={25},
	number={11},
	pages={1435--1441},
	year={1987}
}

@article{bensoussan2021mathematical,
	title={Mathematical formulation of a dynamical system with dry friction subjected to external forces},
	author={Bensoussan, A. and Brouste, A. and Cartiaux, F. B. and Mathey, C. and Mertz, L.},
	journal={Physica D: Nonlinear Phenomena},
	volume={421},
	pages={132859},
	year={2021},
	publisher={Elsevier}
}

@article{roy2005review,
	title={Review of code and solution verification procedures for computational simulation},
	author={Roy, C. J.},
	journal={Journal of Computational Physics},
	volume={205},
	number={1},
	pages={131--156},
	year={2005},
	publisher={Elsevier}
}

@article{mertz2024exponential,
	title={Exponential mixing of constrained random dynamical systems via controllability conditions},
	author={Mertz, L. and Nersesyan, V. and Rissel, M.},
	journal={SIAM Journal on Control and Optimization},
	volume={62},
	number={6},
	pages={3266--3287},
	year={2024},
	publisher={SIAM}
}

@article{bensoussan2012long,
	title={Long cycle behavior of the plastic deformation of an elasto-perfectly-plastic oscillator with noise},
	author={Bensoussan, A. and Mertz, L. and Yam, S. C. P.},
	journal={Comptes Rendus. Math{\'e}matique},
	volume={350},
	number={17-18},
	pages={853--859},
	year={2012}
}

@article{mertz2015degenerate,
	title={Degenerate {D}irichlet problems related to the ergodic property of an elasto-plastic oscillator excited by a filtered white noise},
	author={Mertz, L. and Bensoussan, A.},
	journal={IMA Journal of Applied Mathematics},
	volume={80},
	number={5},
	pages={1387--1408},
	year={2015},
	publisher={Oxford University Press}
}

@article{lauriere2019penalization,
	title={Penalization of nonsmooth dynamical systems with noise: ergodicity and asymptotic formulae for threshold crossings probabilities},
	author={Lauri{\`e}re, M. and Mertz, L.},
	journal={SIAM Journal on Applied Dynamical Systems},
	volume={18},
	number={2},
	pages={853--880},
	year={2019},
	publisher={SIAM}
}

@article{mertz2019numerical,
	title={Numerical analysis of degenerate {K}olmogorov equations of constrained stochastic {H}amiltonian systems},
	author={Mertz, L. and Pironneau, O.},
	journal={Computers and Mathematics with Applications},
	volume={78},
	number={8},
	pages={2719--2733},
	year={2019},
	publisher={Elsevier}
}

@article{lauriere2019free,
	title={Free boundary value problems and {HJB} equations for the stochastic optimal control of elasto-plastic oscillators},
	author={Lauri{\`e}re, M. and Li, Z. and Mertz, L. and Wylie, J. and Zuo, S.},
	journal={ESAIM: Proceedings and Surveys},
	volume={65},
	pages={425--444},
	year={2019},
	publisher={EDP Sciences}
}

@article{lauriere2015penalization,
	title={Penalization of a stochastic variational inequality modeling an elasto-plastic problem with noise},
	author={Lauri{\`e}re, M. and Mertz, L.},
	journal={ESAIM: Proceedings and Surveys},
	volume={48},
	pages={226--247},
	year={2015},
	publisher={EDP Sciences}
}

@article{ip2025control,
	title={A control variate method for threshold crossing probabilities of plastic deformation driven by transient coloured noise},
	author={Ip, H. L. F. and Mathey, C. and Mertz, L. and Wylie, J.},
	journal={Communications in Nonlinear Science and Numerical Simulation},
	year={2025}
}

@article{qiu2020note,
	title={A note on the {M}onge--{A}mp{\`e}re type equations with general source terms},
	author={Qiu, W. and Tang, L.},
	journal={Mathematics of Computation},
	volume={89},
	number={326},
	pages={2675--2706},
	year={2020}
}

@article{zhong2024analysis,
	title={Analysis of a narrow-stencil finite difference method for approximating viscosity solutions of fully nonlinear second order parabolic {PDE}s},
	author={Zhong, X. and Qiu, W.},
	journal={Journal of Scientific Computing},
	volume={99},
	number={3},
	pages={76},
	year={2024},
	publisher={Springer}
}

@article{wu2023unconditionally,
	title={An unconditionally stable and {$L^2$} optimal quadratic finite volume scheme over triangular meshes for anisotropic elliptic equations},
	author={Wu, X. and Qiu, W. and Pan, K.},
	journal={Advances in Computational Mathematics},
	volume={49},
	number={6},
	pages={83},
	year={2023},
	publisher={Springer}
}

@article{li2022convergent,
	title={A convergent post-processed discontinuous {G}alerkin method for incompressible flow with variable density},
	author={Li, B. and Qiu, W. and Yang, Z.},
	journal={Journal of Scientific Computing},
	volume={91},
	number={1},
	pages={2},
	year={2022},
	publisher={Springer}
}

@article{qiu2023global,
	title={Global {$W^{2,p}$} estimates for elliptic equations in the non-divergence form},
	author={Qiu, W. and Tang, L.},
	journal={Proceedings of the American Mathematical Society},
	volume={151},
	number={02},
	pages={763--770},
	year={2023}
}

@article{gao2021pointwise,
	title={The pointwise stabilities of piecewise linear finite element method on non-obtuse tetrahedral meshes of nonconvex polyhedra},
	author={Gao, H. and Qiu, W.},
	journal={Journal of Scientific Computing},
	volume={87},
	number={2},
	pages={53},
	year={2021},
	publisher={Springer}
}

\end{document}